\documentclass{article}

\usepackage{fullpage}
\usepackage{color}
\usepackage{amsthm}
\usepackage{amsmath}
\usepackage{mathtools}
\usepackage{mathrsfs}
\usepackage{amssymb}
\usepackage{tikz}

\title{
	Static Symmetry Breaking in Many-Sorted Finite Model Finding\\
	\vspace{0.5cm}
	\large David R. Cheriton School of Computer Science\\
	University of Waterloo\\
	Undergraduate Thesis
}
\author{Joseph Poremba \\[1cm] Advisor: Nancy Day}
\date{August 14, 2019}

\newtheorem{theorem}{Theorem}
\newtheorem{lemma}[theorem]{Lemma}
\newtheorem{proposition}[theorem]{Proposition}
\newtheorem{corollary}[theorem]{Corollary}

\theoremstyle{definition}
\newtheorem{definition}[theorem]{Definition}

\theoremstyle{Example}

\usepackage{thmtools}
\usepackage{thm-restate}
\declaretheorem[name=Theorem,sibling=theorem]{restatable-theorem}

\newcommand{\U}{\mathcal U}
\newcommand{\F}{\mathcal F}
\newcommand{\R}{\mathcal R}

\DeclareMathOperator{\bounds}{bounds}
\DeclareMathOperator{\True}{True}
\DeclareMathOperator{\Bool}{Bool}
\DeclareMathOperator{\Bindings}{Bindings}
\DeclareMathOperator{\SolSym}{SolSym}
\DeclareMathOperator{\ConSym}{ConSym}
\DeclareMathOperator{\DomPerm}{DomPerm}
\DeclareMathOperator{\DomSym}{DomSym}
\DeclareMathOperator{\ConDomSym}{ConDomSym}

\DeclareMathOperator{\FCSP}{\F CSP}
\DeclareMathOperator{\Orb}{Orb}

\DeclareMathOperator{\id}{id}

\newcommand{\restrict}[2]{\left.{#1}\right|_{#2}}

\DeclarePairedDelimiter{\set}{ \{ }{ \} }

\begin{document}
	
\maketitle

\begin{abstract}
	Symmetry in finite model finding problems of many-sorted first-order logic (MSFOL) can be exploited to reduce the number of interpretations considered during search, thereby improving solver performance.
	In this thesis, we situate symmetry of many-sorted finite model finding (MSFMF) problems in a general framework used for constraint satisfaction problems (CSP).
	We survey and classify existing approaches to symmetry for MSFOL as used in tools such as Paradox.
	We provide new insight into how sorts affect the existence of symmetry and how sort inference can be viewed as a symmetry detection mechanism.
	Finally, we present two new symmetry breaking schemes for MSFOL that are implemented at the MSFOL level and discuss when schemes can be combined. 
	We prove the correctness of our new methods.
\end{abstract}

\section{Introduction}

The satisfiability of first-order logic (FOL) is in general undecidable, but when limited to domains of a finite and small size the search becomes feasible.
According to Jackson's small scope hypothesis \cite{jackson_software_2012}, most flaws in formal models can be exposed by counterexamples of small size, so if all possible interpretations up to a sufficiently high size are explored and no counterexample is found, then we can have a high degree of confidence in the correctness of the model.
Therefore, the process of \emph{finite model finding} (FMF) is an important tool for formal methods.
This importance is evident with the popularity of Jackson's Alloy Analyzer \cite{jackson_software_2012} and its underlying model finder Kodkod \cite{grumberg_kodkod:_2007}, which have found applications in security, e-commerce, feature modeling, program verification, network protocols, and many more fields.

Automated finite model finders generally take one of two main approaches.
In the SEM-style, the solver directly traverses the search space to find a solution, usually through a backtracking algorithm with constraint propagation techniques.
Mace4 \cite{bonacina_mace4_2013}, SEM \cite{goos_system_1996}, and FALCON \cite{zhang_constructing_1996} are examples of SEM-style solvers.
There is also the MACE-style, in which a solver does not directly search for a solution but instead encodes the problem in another logic and invokes an external solver.
Kodkod \cite{grumberg_kodkod:_2007} and Paradox \cite{claessen_new_2003} are MACE-style solvers that reduce FMF to SAT, while Fortress \cite{fitzgerald_finite_2016} is a MACE-style solver that instead reduces the problem to the logic of equality with uninterpreted functions (EUF) and invokes an SMT solver.

While decidable, FMF is NP-complete, and so powerful techniques are needed to make the search feasible.
The exploitation of \emph{symmetry} is one such technique.
Consider the problem of finding a Latin Square.
A Latin Square is a $n \times n$ grid of cells.
Each cell contains one of the numbers $1, \dotsc, n$, and no row or column can contain the same number twice.
Given a Latin Square, the numbers $1, \dotsc, n$ can be permuted within it and another Latin Square is obtained.
While technically a different solution, these two Latin Squares appear to be \emph{symmetric} or \emph{isomorphic}, equivalent in structure but merely having different labels.
Similarly, if given a non-solution to the Latin Squares problem, permuting the numbers $1, \dotsc, n$ within the non-solution will also yield a non-solution.

A na\"ive search over all possible interpretations will encounter many isomorphic interpretations.
However, only a single interpretation from an isomorphism class needs to be examined to accept the entire class as solutions or reject the entire class as non-solutions.
Finite model finders take advantage of these symmetries to perform \emph{symmetry breaking} and reduce the search space.
Such symmetry breaking can be done \emph{dynamically} during a search, or \emph{statically} by adding constraints that eliminate redundant interpretations before attempting a search.

The goal of this thesis is to investigate symmetry and static symmetry breaking in finite model finding, in particular with regards to many-sorted finite model finding (MSFMF).
To our knowledge, little work has been done to explore how sorts impact the existence of symmetries and the unique advantages a many-sorted setting gives to static symmetry breaking.
The first half of this thesis explores the surprisingly varied definitions of ``symmetry" throughout finite model finding literature as well as literature for constraint satisfaction problems (CSPs) in order to establish how these disparate definitions are related and explore what new insights can be found about the relationship between symmetries and sorts.
The second half of this thesis investigates static symmetry breaking techniques for MSFMF.
We compare existing symmetry breaking schemes, considering justification for their correctness, and propose new symmetry breaking schemes with rigorous proofs of soundness.
The role of sorts in these schemes is given heavy emphasis.
We also explore conditions for when symmetry breaking schemes can be combined.

The contributions of this thesis are as follows.
\begin{enumerate}
	\item A review of symmetry and static symmetry breaking techniques in constraint satisfaction problems and finite model finding. We emphasize the connection to group theory and provide rigorous proofs of correctness for existing finite model finding static symmetry breaking techniques where previously appeals to intuition were relied upon.
	\item Unification of the various definitions of symmetry in finite model-finding and relating them to symmetry in constraint satisfaction problems. We provide a transformation to embed finite model-finding into the CSP symmetry framework proposed by Cohen et al. \cite{cohen_symmetry_2006} while preserving the notion of ``value symmetry" between both problem types, and see how the notions of symmetry in these two worlds relate.
	\item New insight into the effects of a sort system on symmetries, using the above transformation and framework. We demonstrate that the use of many-sorted system does not affect the existence of symmetries, but rather a sort system proves the existence of symmetries and allows them to be more easily identified. Thus, we establish that sort inference is a symmetry detection mechanism.
	\item A new static symmetry breaking technique in the presence of ``domain-range distinct" (DRD) functions. Such functions only exist in a many-sorted system. Rigorous proofs of correctness are provided for DRD functions of arbitrary arity.
	\item A new symmetry breaking strategy for unary predicates. A proof of correctness is provided, and it is discussed why this scheme does not easily apply to higher arity predicates.
	\item Conditions for when symmetry breaking techniques can be combined soundly, proved using a framework of extended FMF. To our knowledge this has not been rigorously done in existing FMF literature.
\end{enumerate}

\subsection{Outline of Thesis}
Section \ref{section-background} provides necessary background information in finite model finding, constraint satisfaction problems, and group theory.
Even if the reader is very familiar with these concepts, we recommend at least reading the definitions to familiarize oneself with our notation, especially for finite model finding.

Sections \ref{section-csp} through \ref{section-sorts} comprise the first half of this thesis that is more theoretical and concerned with formalizing the concept of symmetry and understanding how it arises in CSP and FMF problems.
Section \ref{section-csp} looks at symmetry in CSP and presents a framework developed by Cohen et al. \cite{cohen_symmetry_2006} to formalize symmetry in CSP.
Section \ref{section-fmf} explores symmetry in FMF. In this section we look at existing definitions of symmetry in FMF and use the Value Relabeling Theorem to present a framework to unify these definitions. We also define extended FMF problems and how this affects symmetry and then relate symmetry in FMF and CSP.
Section \ref{section-sorts} investigates sorts, sort inference, and their impact on the existence of symmetries.

The second half of this thesis exploring static symmetry breaking consists of Sections \ref{section-static-sym} through \ref{section-combining}.
Section \ref{section-static-sym} discusses static symmetry breaking and the prevalent techniques used in single-sorted FMF.
Section \ref{section-functions} introduces the concept of ``domain-range distinct" functions and presents a new static symmetry breaking technique that can be applied in their presence.
Section \ref{section-predicates} also introduces a new symmetry breaking scheme, this time for unary predicates.
Finally, Section \ref{section-combining} explores under what conditions the aforementioned schemes can be combined.

\section{Background}
\label{section-background}

We assume the reader has a general familiarity with first-order logic (FOL), in particular the syntax and semantics of MSFOL.
To understand symmetry we must also pull from constraint satisfaction problems (CSPs) and group theory, though we do not assume the reader has prior exposure to these fields.
Unless otherwise noted, all functions are total.

\subsection{Finite Model Finding}

We are concerned with a treatment of \emph{many-sorted} first-order logic (MSFOL), where the domain of discourse is broken into ``sorts" or ``types".
Our definitions are based loosely on Vakili and Day's work \cite{fitzgerald_finite_2016}, though we diverge in our treatments of relations and domains in order to make later definitions and theorems easier to state.
The formulations are equivalent.

\begin{definition}[Signature]
	A \emph{signature} is a triple $\Sigma = (\Theta, \mathscr F, \mathscr R)$ where
	\begin{itemize}
		\item $\Theta = \Theta(\Sigma)$ is a finite set of symbols called \emph{sorts},
		\item $\mathscr F = \mathscr F (\Sigma)$ is a finite set of \emph{functional symbols} of the form $f: A_1 \times \dotsb \times A_n \to B$, where $A_1, \dotsc, A_n, B \in \Theta$, and
		\item $\mathscr R = \mathscr R (\Sigma)$ is a finite set of \emph{relational symbols} (or \emph{predicate symbols}) of the form $R: A_1 \times \dotsb \times A_n \to \Bool$, where $A_1, \dotsc, A_n \in \Theta$. 
	\end{itemize}
	For functional symbol $f: A_1 \times \dotsb \times A_n \to B$ or relational symbol $R: A_1 \times \dotsb \times A_n \to \Bool$, $n$ is called the \emph{arity} of the symbol.
	\emph{Constants} are nullary function symbols; that is, constants are functional symbols with no argument sorts and only an output sort.
	A constant symbol $c$ with output sort $A$ will be written $c: A$.
\end{definition}

There are some important things to note about our formulation.
First we emphasize that functional and relational symbols are not themselves actual functions and relations, they are merely symbols that act as placeholders for them.
Additionally, while some treatments of MSFOL consider $\Bool$ as a sort with some special rules, our formulation does not consider $\Bool$ to be a sort. 
We define \emph{$\Sigma$-terms} and \emph{$\Sigma$-formulas} (usually called simply terms and formulas) in the usual way for MSFOL.
We allow the use of the equality predicate $=$ without it needing to be explicitly in the signature.

Now that we can build formulas over a signature, we need a universe in which to interpret them.
We are only concerned here with \emph{finite} domains.

\begin{definition}[Domain Assignment]
	A \emph{domain assignment} for a signature $\Sigma$ is a function $\U$ mapping each sort $\theta \in \Theta(\Sigma)$ to a non-empty, finite set of values $\U(\theta)$, called the \emph{domain} or \emph{values} of $\theta$.
\end{definition}

While it is common to think about distinct sorts being assigned disjoint sets, this need not be the case.
The sort system prevents there from being any undesired interactions between distinct sorts regardless of whether their underlying domains are disjoint.
We are now ready to define a finite model finding problem.

\begin{definition}[Many-Sorted Finite Model Finding Problem]
	A \emph{many-sorted finite model finding} (MSFMF) problem is a tuple $(\Sigma, \Gamma, \U)$ where
	\begin{itemize}
		\item $\Sigma$ is a signature,
		\item $\Gamma$ is a finite set of MSFOL formulas over $\Sigma$, and
		\item $\U$ is a domain assignment for $\Sigma$.
	\end{itemize}
	Together the elements of the sets assigned by $\U$ are referred to as the \emph{values}.
\end{definition}

\begin{definition}[Interpretation]
	An \emph{interpretation} of a MSFMF problem $P = (\Sigma, \Gamma, \U)$ is a mapping $I$ that
	\begin{itemize}
		\item assigns each functional symbol $f: A_1 \times \dotsb \times A_n \to B$ a function $I(f): \U(A_1) \times \dotsb \U(A_n) \to \U(B)$, and
		\item assigns each relational symbol $R: A_1 \times \dotsb \times A_n \to \Bool$ a relation $I(R) \subseteq \U(A_1) \times \U(A_n)$.
	\end{itemize} 
	Note that $I(f)$ and $I(R)$ are actual functions and relations, not just symbols.
\end{definition}

\begin{definition}[The ``Satisfies" Relation]
	 Let $P = (\Sigma, \Gamma, \U)$ be a MSFMF problem and $I$ an interpretation of $P$.
	 Given a formula $\Phi \in \Gamma$, we say $I$ \emph{satisfies} $\Phi$, written $I \models \Phi$, if $\Phi$ is true when interpreted under $I$ using standard first-order logic semantics.
	 We say $I$ \emph{satisfies} $\Gamma$, written $I \models \Gamma$, if it satisfies every formula in $\Gamma$.
	 If $I$ satisfies $\Gamma$, it is said that $I$ is a \emph{model} or \emph{solution} of $P$.
	 If $P$ has a model, it is \emph{satisfiable}, otherwise it is \emph{unsatisfiable}.
\end{definition}

The goal of finite model finding is to determine whether a given MSFMF problem is satisfiable.
We will discuss later how only the size of the sets in the domain assignment affects the satisfiability of the MSFMF problem, not the specific sets themselves.
That is, if $|\U(\theta)| = |\U'(\theta)|$ for all $\theta \in \Theta(\Sigma)$, then $(\Sigma, \Gamma, \U)$ is satisfiable if and only if $(\Sigma, \Gamma, \U')$ is satisfiable.
Thus in many equivalent formulations of the finite model finding problem there is no domain assignment but simply a number $\bounds(\theta)$ for each sort $\theta$ that specifies the size of the domain for $\theta$.
However, we opt for this more general formulation with domain assignments because it facilitates later theorems and discussions.

\subsection{Constraint Satisfaction Problems}

Constraint programming is a general paradigm for finding solutions to problems involving decision variables subject to a set of constraints.
We base our definition of constraint satisfaction problems (CSPs) on those found in Freuder and Mackworth's work \cite{freuder_constraint_2006}, though we concern ourselves with only finite problems similar to Cohen et. al \cite{cohen_symmetry_2006}.

\begin{definition}[Constraint Satisfaction Problem]
	A \emph{constraint satisfaction problem} (CSP) is a triple $(X, D, C)$ where
	\begin{itemize}
		\item $X$ is a finite set of \emph{variables} $X = \set{x_1, \dotsc, x_n}$,
		\item $D$ is a mapping that maps each variable $x_i \in X$ to a nonempty, finite set $D(x_i)$, called its \emph{domain}, and
		\item $C$ is a finite set of \emph{constraints}.
	\end{itemize}
	A constraint is a pair $C_j = (S_j, R_j)$, where $S_j$ is a sequence of variables $S_j = (y_1, \dotsc, y_k)$, $y_i \in X$, called the \emph{scope}\footnote{
		The word ``scope" is overloaded.
		In constraint programming, it refers to the variables affected by a constraint.
		In finite model finding, it sometimes refers to the size of the set of values in a domain assignment, what we earlier called $\bounds(\theta)$ for a sort $\theta$.
		From this point on, when we say ``scope", we refer to the former and never the latter.}
	of the constraint, and $R_j$ is a relation $R_j \subseteq D(y_1) \times \dotsb \dotsc D(y_k)$ that specifies the allowed tuples for variables in the scope.
\end{definition}

Following Cohen et al. \cite{cohen_symmetry_2006}, we allow a constraint to be defined \emph{extensionally} through an explicit list of allowed tuples, as done above, or \emph{intentionally} by instead giving an expression from which the allowed tuples are determined.
For example, consider a SAT problem which has variables $a$ and $b$.
The constraint that $a$ and $b$ must be the same can be stated extensionally by a constraint $C = (S, J)$ where $S = (a, b)$ and $R = \set{ (T, T), (F, F) }$, but is more commonly written intensionally with the expression $a \iff b$.

We now define how to form a solution to a CSP.

\begin{definition}[Bindings, Assignments]
	For a CSP, $P = (X, D, C)$, a \emph{variable-value binding}, or simply a \emph{binding}, is a tuple $(x, v)$ such that $x \in X$ and $v \in D(x)$.
	We denote  the set of all variable-value bindings of $P$ by $\Bindings(P)$.
	That is, $\Bindings(P) = \set{ (x, v) : x \in X, v \in D(x)}$.
	A set of bindings is a \emph{partial assignment} if it contains at most one binding of the form $(x, v)$ for each $x \in X$.
	A partial assignment is a \emph{complete assignment} if it contains exactly one binding of the form $(x, v)$ for each $x \in X$.
	We can alternatively view a complete assignment as a function mapping variables to values, and so for complete assignment $A$ we use $A(x)$ to denote the value $v$ such that $(x, v) \in A$.
\end{definition}

\begin{definition}[The ``Satisfies" Relation, Solution]
	Let $P = (X, D, C)$ be a CSP.
	A complete assignment $A$ \emph{satisfies} a constraint $C_j$ if $(A(y_1), \dotsc, A(y_k)) \in R_j$, where $(y_1, \dotsc, y_k)$ is the scope of $C_j$ and $R_j$ is the relation of $C_j$.
	A \emph{solution} to $P$ is a complete assignment $A$ that satisfies all of the constraints of $P$.
\end{definition}

For example, consider a CSP $P = (X, D, C)$ where $X = \set{x, y, z}$, $D(x) = \set{1, 2}$, $D(y) = \set{3, 4}$, $D(z) = \set{5, 6}$, and $C$ consists of a single constraint $C_1$ with scope $S_1 = (x, z)$ and relation $R_1 = \set{(1, 6), (2, 5)}$.
Consider the complete assignments $A_1 = \set{(x, 1), (y, 3), (z, 5)}$ and $A_2 = \set{(x, 2), (y, 3), (z, 5)}$.
$A_1$ is not a solution since $(A_1(x), A_1(z)) = (1, 5) \notin R_1$, but $A_2$ is a solution because $(A_2(x), A_2(z)) = (2, 5) \in R_1$.

\subsection{Group Theory}

The intuitive, informal notion of a symmetry is that it is some transformation, like a rotation or reflection, that preserves an important property of an object.
Symmetries can be composed, applied one after another, and still maintain this property.
Therefore a symmetry can be created by composing symmetries together.
Symmetries can also be reversed while still maintaining the desired property of the object.
Hence, taking the inverse of a symmetry also yields a symmetry.

Group theory is the rigorous mathematical formalism used to explore symmetry.
The abstract definition of a group captures the composition and inversion behaviours described above, as well as other important properties.
We take our definitions in this section from Dummit and Foote \cite{dummit_abstract_2014}.

\begin{definition}[Group]
	A \emph{group} is a set $G$ equipped with a binary function $\circ$ from $G \times G$ to $G$ (written infix, called the \emph{group operation}) such that
	\begin{itemize}
		\item the $\circ$ operator is \emph{associative}; that is, $a \circ (b \circ c) = (a \circ b) \circ c$ for all $a, b, c \in G$,
		\item there exists an \emph{identity} element $e \in G$ such that $a \circ e = e \circ a = e$ for all $a \in G$,
		\item each element $a \in G$ has an \emph{inverse} $a^{-1}$ such that $a \circ a^{-1} = a^{-1} \circ a = e$.
	\end{itemize}
	For convenience, the $\circ$ function is often omitted in notation and $a \circ b$ is instead written $ab$.
\end{definition}

The most common example of a group is the \emph{symmetric group} $S_A$, which is the set of permutations of a set $A$.
The group operation is function composition, which is associative.
The identity element is the identity function, and the group inverse of a permutation $\sigma$ is just its standard inverse as a function.
Another common example is the set of non-zero real numbers $\mathbb R^\times$.
The group operation is multiplication, the identity element is $1$, and the group inverse of a non-zero number $a$ is its multiplicative inverse $\frac 1 a$.

A subset of a group might itself be a group under the same operation, in which case it is called a \emph{subgroup}.

\begin{definition}[Subgroup]
	Let $G$ be a group.
	A subset $H$ of $G$ is a \emph{subgroup} of $G$, written $H \le G$, if
	\begin{itemize}
		\item $H$ is non-empty,
		\item $H$ is closed under the group operation, and
		\item $H$ is closed under taking inverses.
	\end{itemize}
\end{definition}
Consider for example an undirected graph $G = (V, E)$.
An \emph{automorphism} of $G$ is a permutation $\pi$ of its vertices such that a pair $(u, v)$ of its vertices forms an edge if and only if $(\pi(u), \pi(v))$ forms an edge.
The set of automorphisms of $G$ is a subgroup of the symmetric group $S_V$ of its vertex set.

For our purposes we are interested in how the group elements act on objects to transform them into other objects, which is formalized by the notion of a \emph{group action}.

\begin{definition}[Group Action]
	A \emph{group action} of a group $G$ on a set $X$ is a binary function $\bullet$ from $G \times X$ to $X$ (written infix) such that
	\begin{enumerate}
		\item $g_1 \bullet (g_2 \bullet x) = (g_1 g_2) \bullet x$ for all $g_1, g_2 \in G$ and $x \in X$, and
		\item $e \bullet x = x$ for all $x \in X$.
	\end{enumerate}
	It is then said that \emph{$G$ acts on $X$} by this group action.
\end{definition}

As an example, consider the permutation group $S_A$ of the set $A = \set{1, \dotsc, n}$.
$S_A$ acts on $A$ through function application.
That is, defining $\sigma \bullet a = \sigma(a)$ for each $\sigma \in S_A$ and $a \in A$ forms a group action.
Any group $G$ acts on any set $X$ by the \emph{trivial action}, defined by $g \bullet x = x$ for all $g \in G$ and $x \in X$.
A more complicated example of a group action is \emph{conjugation}.
Any group $G$ acts on itself by $g \bullet h = g h g^{-1}$ for all $g, h \in G$.

One way, not mentioned in Dummit and Foote, that we find helpful in visualizing group actions for those unfamiliar with them is to view the action as a transition system or edge-labeled directed graph.
The nodes of the transition system are the elements of $X$.
For each $x \in X$ and $g \in G$ there is an arc from $x$ to $g \bullet x$.
The first property of a group action says that to start from $x$ and follow the $g_2$ arc and then the $g_1$ arc, the result is the same as starting from $x$ and instead taking the single $g_1 g_2$ arc.
If the reversed order seems confusing, just note that this is exactly how function composition works.
The first property also implies that starting from $x$ and following the $g$ arc, then the $g^{-1}$ arc, results in returning to $x$.

Going back to Dummit and Foote, it is not hard to show that each individual group element can be viewed as a permutation of the set $X$ under the group action (and hence, group actions fundamentally link groups to permutations and symmetric groups).
This is called the \emph{permutation representation} of the group action.

\begin{proposition}[Permutation Representation of a Group Action]
	\label{prop-permutation-representation}
	Let $G$ be a group acting on a set $X$.
	For every $g \in G$, the function $\sigma_g: X \to X$ defined by $\sigma_g(x) = g \bullet x$ is a permutation of $X$. 
\end{proposition}

Next want to consider all the objects that we can get to from an $x \in X$ by acting on it with group elements.

\begin{definition}
	Let $G$ be a group acting on a set $X$.
	For each $x \in X$, define the \emph{orbit} of $x$ by $\Orb_x = \set{ g \bullet x : g \in G}$.
\end{definition}

Since traveling along an arc can be reversed by traveling backwards through the inverse group element, it is immediate that $y \in \Orb_x$ if and only if $x \in \Orb_y$.
Being related in this way forms an equivalence relation and so partitions $X$ into equivalence classes by the orbits.

\begin{proposition}
	Let $G$ be a group acting on a set $X$.
	The relation $\sim$ on $X$ defined by
	\[ x \sim y \text{ if and only if } x \in \Orb_y \]
	is an equivalence relation.
	The equivalence classes are the orbits, and hence the orbits partition $X$.
\end{proposition}

It is also worth noting that if a group $G$ acts on a set $X$, any subgroup $H \le G$ acts on $X$ by the same group action.
This proposition is immediate from the definitions of subgroup and group action.

\begin{proposition}[Subgroups and Group Actions]
	If $G$ is a group acting on a set $X$, then any subgroup $H \le G$ acts on $X$ by the same group action.
\end{proposition}

\section{The CSP Symmetry Framework}
\label{section-csp}

Before symmetry can actually be exploited, we must clearly define what exactly we mean by symmetry.
However, as has been noted in the works of Gent et al. \cite{gent_symmetry_2006} and Cohen et al. \cite{cohen_symmetry_2006}, definitions of symmetry across CSP research have been surprisingly disparate.
We refer the reader to those sources for a more in-depth survey.
Here we briefly summarize their findings.
The definitions of symmetry generally tend to disagree along one of the following lines:
\begin{enumerate}
	\item whether symmetries are defined as acting on variables, values, or variable-value bindings, and
	\item whether symmetries are defined as preserving the set of \emph{solutions} or the set of \emph{constraints}.
\end{enumerate}

To understand these different notions of symmetry, consider two examples of CSP problems: SAT and Sudoku.
In a SAT problem, two variables might be interchangeable.
For example if the constraints are $a \implies b$ and $c \implies b$, $a$ and $c$ would be interchangeable and we would consider a permutation on the variables that swaps $a$ and $c$ to be a symmetry.
However, the values $T$ and $F$ are not in general interchangeable, so permuting values is not considered in the definition of symmetry in the context of SAT.
Contrast this with Sudoku.
In Sudoku, the variables are the grid cells and the values are $\set{1, \dotsc, 9}$.
There is nothing special about the specific choice of values assigned to the cell variables.
For example, relabeling all $1$s in a solution with $9$s and vice versa always generates a solution and so we would consider a permutation on the values that swaps $1$ and $9$ to be a symmetry.
Thus, permuting values would be considered important in the definition of symmetry for Sudoku problems, in contrast with SAT.
More general definitions would consider acting on variable-value bindings.

After having decided whether symmetries are permutations of variables, values, or bindings, the question arises as to what property causes such a permutation to be a symmetry.
The more general notion is that a symmetry preserves the set of solutions.
That is, when it is applied to a solution it always yields another solution, and when it is applied to a non-solution it always yields a non-solution.
The common criticism of this ``solution symmetry" definition is that it provides no guidance on how to practically find such symmetries; finding the full set of symmetries might require finding all of the solutions.
``Constraint symmetries" on the other hand are defined as preserving the set of constraints.
In the SAT problem above, swapping $a$ and $c$ leaves the exact same set of constraints, so such a permutation would be considered a symmetry.
This disagreement in definitions is essentially over whether the fact that symmetries preserve solutions is their defining property or a consequence of preserving constraints.

Cohen et al. \cite{cohen_symmetry_2006} provide a framework to unify these various definitions.
This section gives a presentation of Cohen et al.'s framework, with some extensions to make it slightly more general to facilitate discussion of sorts later.
This presentation is mostly not our original work; the extensions to Cohen et al.'s work here are quite natural.
For example, we do not restrict the framework to a setting where the domains of each variable are all the same as they did, but it is not difficult to remove this restriction.
However, we do endeavour to provide greater intuition and explanations, including some more precise definitions and additional proofs.
The explanations are largely original.
In particular, Cohen et al. do not deeply elaborate on the underlying group theory, so we interleave more discussions of the group theory between the definitions.

\subsection{Solution Symmetries}

The approach taken by Cohen et al. \cite{cohen_symmetry_2006} is to define a very general group (in the group theory sense) of symmetries, and then view other kinds of symmetries as special cases (in particular, subgroups).
So the definition of symmetry starts more abstract, and can be refined as needed.
Their most abstract definition of symmetry acts on variable-value bindings and preserves the set of solutions.
That is, a symmetry is a permutation on $\Bindings(P)$, and it must map solutions to solutions.
Variable and value symmetries are easily viewed as special cases.
We denote by $S_{\Bindings(P)}$ the set of all permutations $\sigma: \Bindings(P) \to \Bindings(P)$, which forms a group.
The objects we are interested in this group acting on are \emph{sets} of bindings (for example, partial and complete assignments).
That is, the space we wish to understand how symmetries act on is $2^{\Bindings(P)}$, 
We must first be clear how a permutation $\sigma: \Bindings(P) \to \Bindings(P)$ acts on a set of bindings.

\begin{definition}[Action of a Permutation on a Set of Bindings]
	For a CSP instance $P$, define the action of a permutation $\sigma: \Bindings(P) \to \Bindings(P)$ on a set of bindings by lifting $\sigma$ pointwise.
	Precisely, we define a binary operator $\bullet$ from $S_{\Bindings(P)} \times 2^{\Bindings(P)}$ to $2^{\Bindings(P)}$ by
	\[ \sigma \bullet \set{ (x_1, v_1), (x_2, v_2), \dotsc, (x_k, v_k) } = \set{ \sigma(x_1, v_1), \sigma(x_2, v_2), \dotsc, \sigma(x_k, v_k)} . \]
\end{definition}

\begin{theorem}
	For a CSP instance $P$, the action of $S_{\Bindings(P)}$ on the set $2^{\Bindings(P)}$ is a group action.
\end{theorem}
\begin{proof}
	Clearly the identity function $\sigma_{\id} \in S_{\Bindings(P)}$ acts on each set $Z \subseteq \Bindings(P)$ by sending $Z$ to itself, so the second requirement of group actions is satisfied.
	
	Now let $\sigma, \pi: \Bindings(P) \to \Bindings(P)$ and $Z \subseteq \Bindings(P)$.
	Write $Z = \set{(x_1, v_1), \dotsc, (x_k, v_k)}$ for some $k \ge 0$.
	We see that
	\[
		\sigma \bullet (\pi \bullet Z)
		= \sigma \bullet \set{ \pi(x_1, v_1), \dotsc, \pi(x_k, v_k) }
		= \set{ \sigma(\pi(x_1, v_1)), \dotsc, \sigma(\pi(x_k, v_k)) }
		= (\sigma \circ \pi) \bullet Z,
	\]
	and hence the first requirement of group actions is satisfied.
\end{proof}

Now that we have defined how permutations of $\Bindings(P)$ act, we define what makes such a permutation a symmetry.

\begin{definition}[Solution Symmetry]
	A \emph{solution symmetry} of a CSP instance $P$ is a permutation $\sigma$ of $\Bindings(P)$ that map solutions of $P$ to solutions of $P$.
	That is, if $Z \subseteq \Bindings(P)$ is a solution of $P$, then $\sigma \bullet Z$ is also a solution of $P$.
	We will denote the set of all solution symmetries of $P$ by $\SolSym(P)$.
\end{definition}

We also want symmetries to map non-solutions to non-solutions.
It may not be obvious to the reader and the sources we have read neglect to explain this, but this is actually a consequence of the definition.
We use the following lemma, though we omit its proof.

\begin{lemma}
	\label{lemma-closed-permutation}
	Let $Z$ be a finite set, $U \subseteq Z$, and $\sigma$ a permutation of $Z$.
	If $U$ is closed under $\sigma$, then $Z \setminus U$ is also closed under $\sigma$.
\end{lemma}

We get the following corollary.

\begin{corollary}
	A solution symmetry $\sigma$ of a CSP instance $P$ maps non-solutions to non-solutions.
	That is, if $Z \subseteq \Bindings(P)$ is not a solution of $P$, then $\sigma \bullet Z$ is not a solution of $P$.
\end{corollary}
\begin{proof}
	Consider the function $f_\sigma : 2^{\Bindings(P)} \to 2^{\Bindings(P)}$ defined by $f_\sigma(Z) = \sigma \bullet Z$ for each $Z \subseteq \Bindings(P)$.
	By Proposition \ref{prop-permutation-representation}, $f_\sigma$ is a permutation on $2^{\Bindings(P)}$.
	The definition of a solution symmetry implies that the set of solutions to $P$ must be closed under $f_\sigma$.
	By Lemma \ref{lemma-closed-permutation}, the set of non-solutions in $2^{\Bindings(P)}$ must also be closed under $f_\sigma$.
	Therefore, if $Z \subseteq \Bindings(P)$ is not a solution of $P$, then $\sigma \bullet Z$ is not a solution of $P$.
\end{proof}

Knowing this now allows us to prove that the set of solution symmetries forms a group.

\begin{proposition}
	For a CSP $P$, $\SolSym(P)$ is a subgroup of the group $S_{\Bindings(P)}$.
\end{proposition}
\begin{proof}
	The identity permutation lies in $\SolSym(P)$, so it is non-empty.
	
	Let $\sigma_1, \sigma_2 \in \SolSym(P)$, and let $Z \subseteq \Bindings(P)$ be a solution to $P$.
	Since $\sigma_2$ map solutions to solutions, $\sigma_2 \bullet Z$ is a solution.
	Furthermore since $\sigma_1$ maps solutions to solutions, $\sigma_1 \bullet (\sigma_2 \bullet Z) = (\sigma_1 \circ \sigma_2) \bullet Z$ is a solution.
	Therefore $\sigma_1 \circ \sigma_2$ maps solutions to solutions, and so $\sigma_1 \circ \sigma_2 \in \SolSym(P)$.
	That is, $\SolSym(P)$ is closed under the group operation of function composition.
	
	Suppose for a contradiction that for some solution symmetry $\sigma$, its inverse $\sigma^{-1}$ is not a solution symmetry.
	Then there exists a solution $Z \subseteq \Bindings(P)$ such that $\sigma^{-1} \bullet Z$ is not a solution.
	Since $\sigma$ maps non-solutions to non-solutions, $\sigma \bullet (\sigma^{-1} \bullet Z)$ is not a solution.
	However $\sigma \bullet (\sigma^{-1} \bullet Z)  = (\sigma \circ \sigma^{-1}) \bullet Z = Z$, so this contradicts that $Z$ is a solution.
	Therefore $\SolSym(P)$ is closed under taking inverses.
\end{proof}

Now that we have a group of solution symmetries and a group action of this group on sets of bindings, we can consider the orbits (i.e. equivalence classes) induced by this action.
If two sets of bindings $Z$ and $Z'$ lie in the same orbit, it means that there must exist a solution symmetry $\sigma$ that acts on $Z$ to produce $Z'$.
Because each solution symmetry preserves whether a set of bindings is a solution, $Z$ is a solution if and only if $Z'$ is a solution.
Hence we have arrived at the following proposition.

\begin{proposition}
	Let $P$ be a CSP instance.
	If $H$ is an orbit produced by the action of $\SolSym(P)$ on $2^{\Bindings(P)}$, either all binding-sets $Z \in H$ are solutions, or none of them are solutions.
\end{proposition}

These orbits are precisely the equivalence or \emph{isomorphism} classes of binding-sets that we are interested in when exploring symmetry.
The key insight in symmetry reduction is that to determine whether a solution exists to a CSP instance, it is necessary only to examine a single representative from each class, which can drastically reduce search time.
More on symmetry reduction will be discussed in later sections.

Note that a symmetry always maps solutions to solutions, but it may not always map complete assignments to complete assignments.

\subsection{Constraint Symmetries}

As Gent et al. \cite{gent_symmetry_2006} remark, defining symmetries as mappings that preserve solutions does not seem to give any guidance on how such symmetries might be identified.
To find the whole symmetry group, one might need to compute all solutions to the CSP, which defeats the purpose in taking advantage of symmetries.
In order for any definition of symmetry to be useful, it must obviously require that solutions are mapped to solutions.
As mentioned earlier, many other CSP definitions do not make this the defining property of symmetries, but instead look at a more structured group for which this property is a consequence of the definition.
In particular, many define a symmetry as a mapping that somehow preserves the set of constraints.

Cohen et al. \cite{cohen_symmetry_2006} recognized that these ``constraint" or ``problem" symmetries can simply be viewed as a subgroup of the solution symmetries.
Here we will present their formalization of this result.
To formalize the notion of preserving constraints, they define a hypergraph called the \emph{microstructure} of a CSP instance, pulled from earlier work by Freuder \cite{freuder1991eliminating} and J{\'e}gou \cite{jegou1993decomposition}.
We only need the complement of this hypergraph, so we will provide the definition for its complement directly.
In this hypergraph, the edges connect bindings that cannot together be part of a solution.

\begin{definition}[Microstructure Complement]
	Let $P = (X, D, C)$ be a CSP.
	The \emph{microstructure complement} of $P$ is the hypegraph $\overline{MS}(P) = (V, E)$, where $V = \Bindings(P)$ and $E$ contains precisely the following hyperedges:
	\begin{itemize}
		\item $E$ contains the hyperedge $e = \set{(x, a), (x, b)}$ for each $x \in X$ and distinct $a, b \in D(x)$, and
		\item $E$ contains the hyperedge $e = \set{ (x_1, a_1), \dotsc, (x_k, a_k) }$ if $\set{x_1, \dotsc, x_k}$ is the set of variables of some constraint scope, but that constraint disallows the assignment $e$.
	\end{itemize}
	We will call hyperedges of the first type \emph{consistency hyperedges}, and those of the latter type \emph{constraint hyperedges}.
\end{definition}

Consistency hyperedges reflect that assignments must be internally consistent; that is, one cannot simultaneously assign variable $x$ to have value $a$ and value $b$.
Constraint hyperedges reflect that assignments must satisfy the constraints.
An \emph{independent set} in a hypergraph is a set $Y$ of vertices such that none of the hyperedges is a subset of $Y$ ($Y$ is allowed to intersect a hyperedge, but it cannot entirely contain one).
From the definition, we can see that a solution to the CSP is exactly an independent set of size $|X|$.

Consider the SAT problem on variables $X = \set{a, b, c}$ with the following constraints.
\[C_1 \coloneqq a \land b \]
\[C_2 \coloneqq a \lor b \lor c \]
The microstructure complement of this problem is presented below.
The green hyperedges are the consistency hyperedges.
The blue and red hyperedges are constraint hyperedges, arising from $C_1$ and $C_2$ respectively.
The independent set $J = \set{(a, T), (b, T), (c, F)}$, which satisfies that $|J| = |X|$, is a solution to the SAT problem.

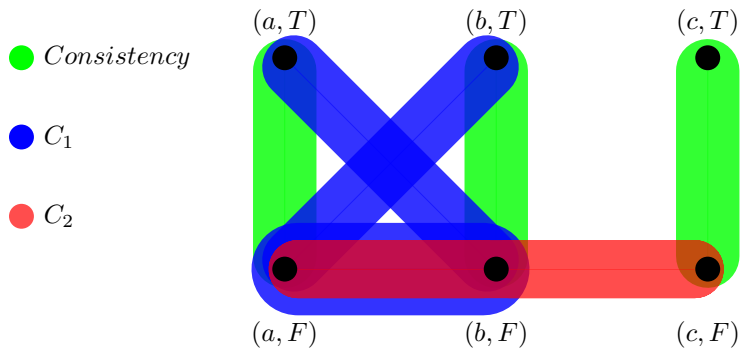
\tikzstyle{vertex} = [fill,shape=circle,node distance=80pt]
\tikzstyle{edge} = [fill,opacity=.8,fill opacity=.5,line cap=round, line join=round, line width=30pt]
\tikzstyle{elabel} =  [fill,shape=circle,node distance=30pt]

\pgfdeclarelayer{background}
\pgfsetlayers{background,main}

\begin{center}

\begin{tikzpicture}[scale = 0.7]
\node[vertex,label=above:\({(a, T)}\)] (v1) {};
\node[vertex,right of=v1,label=above:\({(b, T)}\)] (v2) {};
\node[vertex,right of=v2,label=above:\({(c, T)}\)] (v3) {};
\node[vertex,below of=v1,label={[label distance=0.4cm]below:\({(a, F)}\)}] (v4) {};
\node[vertex,right of=v4,label={[label distance=0.4cm]below:\({(b, F)}\)}] (v5) {};
\node[vertex,right of=v5,label={[label distance=0.4cm]below:\({(c, F)}\)}] (v6) {};

\begin{pgfonlayer}{background}

\draw[edge,color=green, line width=24pt] (v1) -- (v4);
\draw[edge,color=green, line width=24pt] (v2) -- (v5);
\draw[edge,color=green, line width=24pt] (v3) -- (v6);

\draw[edge,color=blue, line width=24pt] (v4) -- (v2);
\draw[edge,color=blue, line width=35pt] (v4) -- (v5);
\draw[edge,color=blue, line width=24pt] (v1) -- (v5);

\draw[edge,color=red,opacity=0.7pt, line width = 22pt] (v4) -- (v5) -- (v6) -- (v4);

\end{pgfonlayer}

\node[elabel,color=green,label=right:\({Consistency}\)]  (e0) at (-5,0) {};
\node[elabel,below of=e0,color=blue,label=right:\({C_1}\)]  (e1) {};
\node[elabel,below of=e1,color=red,opacity=0.7pt,label=right:\({C_2}\)]  (e2) {};
\end{tikzpicture}

\end{center}

Consistency hyperedges can be viewed as a kind of implicit constraint that a variable cannot be assigned multiple values at once, and each remaining problem constraint clearly corresponds to a set of constraint hyperedges.
So the notion of ``preserving constraints" can be viewed as preserving hyperedges of the microstructure complement.
Recalling that an \emph{automorphism} of a hypergraph is a permutation of its vertices that, when extended pointwise to sets of vertices, maps hyperedges to hyperedges and non-hyperedges to non-hyperedges, Cohen et al. define the constraint symmetries as follows.

\begin{definition}[Constraint Symmetry]
	A \emph{constraint symmetry} of a CSP instance $P$ is an automorphism of $\overline{MS}(P)$.
	We will denote the set of all constraint symmetries of $P$ by $\ConSym(P)$.
\end{definition}

Next they prove that the constraint symmetries of a CSP instance are a subgroup of the solution symmetries.

\begin{theorem}
	For a CSP instance $P = (X, D, V)$, $\ConSym(P)$ is a subgroup of $\SolSym(P)$.
\end{theorem}
\begin{proof}
	The automorphisms of a hypergraph always form a group under the operation of function composition, so it suffices to show that $\ConSym(P) \subseteq \SolSym(P)$.
	Let $\sigma \in \ConSym(P)$.
	Suppose $Z \subseteq \Bindings(P)$ is a solution.
	Then in the microstructure complement, $Z$ is an independent set of size $|X|$.
	Since $\sigma$ sends non-hyperedges to non-hyperedges, $\sigma \bullet Z$ is an independent set of size $|X|$ and hence is a solution.
	Thus $\sigma$ always maps solutions to solutions and $\sigma \in \SolSym(P)$.
\end{proof}

Cohen et al. also gave an example of a CSP where this containment is strict; that is, there is a solution symmetry that is not a constraint symmetry.

\section{Symmetry in Finite Model Finding}
\label{section-fmf}

We now explore notions of symmetry as they appear in finite model finding, unify these concepts, and relate them to the CSP symmetry framework.

\subsection{Existing Symmetry Definitions in FMF}

In the literature for finite model finding there likewise are many definitions of symmetry, though they tend to be more consistent with each other than those found in CSP literature.
All those we have surveyed are interested in a notion of value symmetry, noting that in finite model finding all values are interchangeable. 
Relabeling the values in an interpretation leads to another interpretation that satisfies the formulas if and only if the original interpretation also satisfies the formulas.

In their introduction to the Falcon model finder, Zhang \cite{zhang_constructing_1996} considers expanding FOL formulas into ground equations. They then define a set of ground equations as symmetric with respect to some set $X$ of values if permutating the values of $X$ within the ground equations does not change the set of ground equations.
Peltier \cite{peltier_new_1998}, and later Claessen and S\"orenson when introducing their Paradox model finder \cite{claessen_new_2003}, define isomorphic interpretations that can be created by permuting (relabeling) the domain values.
Torlak's work \cite{torlak_constraint_2009} is in an FMF setting where not all value permutations preserve solutions, and so defines them to be specifically those that do preserve solutions.
Baumgartner et al. \cite{baumgartner_computing_2009} similarly define a value symmetry to be a permutation of the values, but one that maps partial solutions to partial solutions.
Audemard and Benhamou \cite{audemard2001symmetry} consider a multi-sorted setting and define symmetry in terms of a dynamic search.
For them, a symmetry is a collection of permutations, one for each sort's domain, that leaves unchanged both the current instantiation and the set of ground clauses. 

The papers by Zhang and Zhang that introduce their SEM tool \cite{goos_system_1996} and McCune's Mace4 solver \cite{bonacina_mace4_2013} do not define symmetry but base their symmetry exploiting techniques on Zhang's previous work for Falcon.
Similarly, Vakili and Day's paper \cite{fitzgerald_finite_2016}, introducing the Fortress solver, does not define symmetry but uses the same symmetry breaking techniques found in Paradox.
The same is true for the paper by Reger et al. \cite{creignou_finding_2016} on their work in finite model finding using the Vampire theorem prover.

As stated earlier, these definitions are mostly consistent, at least in spirit.
However, there are several limitations that we can identify.
First, the notions of symmetry are restricted only to values; unlike in CSP, there is no concept of symmetries acting on variables or more generally variable-value pairs.
Second, most definitions, barring those of Peltier and Audemard and Benhamou, only rigorously define symmetries for single-sorted problems.
Additionally, as was the case in CSP, there is not a consensus as to whether symmetries should be defined as preserving solutions or constraints.
Torlak defines them as preserving solutions, Zhang and Audemard and Benhamou define them as preserving constraints, and others consider only restricted settings where the distinction need not be addressed.
Finally, while there are many conceptual similarities between symmetry in CSP and symmetry in finite model finding, and no doubt each field inspires work in the other, there does not appear to be an attempt to rigorously tie the two together.
In this section, we propose to use the work of Cohen et al. \cite{cohen_symmetry_2006} to unify the symmetry definitions for finite model-finding and CSP.

\subsection{Value Relabeling}

Earlier in this thesis, we mentioned that only the \emph{sizes} of the sets given by a domain assignment affect the satisfiability of the problem, the actual sets themselves do not matter.
What the values actually are does not factor into the semantics of FOL; it only matters that the values are distinct from each other.
For example, it does not matter whether we say that $\U(A) = \set{1, 2, 3, 4}$ or $\U(A) = \set{a, b, c, d}$ or $\U(A) = \set{\text{cat}, \text{dog}, \text{mouse}, \text{rabbit}}$.
The values can be relabeled without affecting satisfiability.
This well-known property of FOL forms the basis of the current understanding of symmetry in FOL.
The following definitions and theorem express this relabeling phenomenon formally.
This theorem can be found in the works by Claessen and S\"orenson \cite{claessen_new_2003} and Peltier \cite{peltier_new_1998}.

\begin{definition}[Domain Bijection]
	Let $\Sigma$ be a signature.
	Given two domain assignments $\U$ and $\U'$ for $\Sigma$, a \emph{domain bijection} between $\U$ and $\U'$ is a collection of bijections, containing a bijection $\sigma_\theta : \U(\theta) \to \U'(\theta)$ for each $\theta \in \Theta(\Sigma)$.
\end{definition}

\begin{definition}[Action of a Domain Bijection on an Interpretation]
	Let $\Sigma$ be a signature, $\Gamma$ a set of formulas over $\Sigma$, $\U$ and $\U'$ two domain assignments for $\Sigma$, and $\sigma$ a domain bijection between $\U$ and $\U'$.
	Let $I$ be an interpretation for $(\Sigma, \Gamma, \U)$.
	We define the action of $\sigma$ on $I$ as producing an interpretation $\sigma \bullet I$ for $(\Sigma, \Gamma, \U')$, constructed as follows.
	
	For each functional symbol $f: A_1 \times \dotsb \times A_n \to B$, define $(\sigma \bullet I)(f)$ by
	\[ (\sigma \bullet I)(f)(a'_1, \dotsc, a'_n) = \sigma_B( I(f)( \sigma_{A_1}^{-1}(a'_1), \dotsc, \sigma_{A_n}^{-1}(a'_n) ) ) .\]
	Equivalently, if $I(f)(a_1, \dotsc, a_n) = b$, then $(\sigma \bullet I)(f)( \sigma_{A_1}(a_1), \dotsc, \sigma_{A_n}(a_n) ) = \sigma_B(b)$.
	
	For each relational symbol $R: A_1 \times \dotsb A_n \to \Bool$, define $(\sigma \bullet I)(R)$ by
	\[ (a_1, \dotsc, a_n) \in I(R) \iff (\sigma_{A_1}(a_1), \dotsc, \sigma_{A_n}(a_n)) \in (\sigma \bullet I)(R) .\]
\end{definition}

\begin{theorem}[Value Relabeling]
	\label{value-relabelling}
	Let $\Sigma$ be a signature, $\Gamma$ a set of formulas over $\Sigma$, $\U$ and $\U'$ two domain assignments for $\Sigma$, and $\sigma$ a domain bijection between $\U$ and $\U'$.
	For any interpretation $I$ of $(\Sigma, \Gamma, \U)$, $I \models \Gamma$ if and only if $(\sigma \bullet I) \models \Gamma$.
	Hence, $I$ is a model of $(\Sigma, \Gamma, \U)$ if and only if $(\sigma \bullet I)$ is a model of $(\Sigma, \Gamma, \U')$.
\end{theorem}

As an example, consider $\Sigma, \Gamma, \U, \U'$ defined below.
\begin{align*}
&\Theta(\Sigma) = \set{A, B}, \mathscr F(\Sigma) = \set{c: A, d: B, f: A \to B}, \mathscr R(\Sigma) = \set{P: A \times B \to \Bool} \\
&\Gamma = \set{\forall x : B . \, P(c, x), \exists y : A . \, f(y) = d } \\
&\U(A) = \set{a_1, a_2, a_3}, \U(B) = \set{b_1, b_2} \\
&\U'(A) = \set{\alpha, \beta, \omega}, \U'(B) = \set{\star, \diamond}
\end{align*}
One possible domain bijection $\sigma$ is given by the following.
\begin{align*}
&\sigma_A(a_1) = \alpha, \sigma_A(a_2) = \beta, \sigma_A(a_3) = \omega \\
&\sigma_B(b_1) = \star, \sigma_B(b_2) = \diamond
\end{align*}
Now define an interpretation $I$ for $(\Sigma, \Gamma, \U)$ as follows.
\begin{align*}
&I(c) = a_1\\
&I(d) = b_2\\
&I(f)(a_1) = b_2, I(f)(a_2) = b_2, I(f)(a_3) = b_2\\
&I(P) = \set{(a_1, b_1), (a_1, b_2), (a_3, b_2)}
\end{align*}
$I$ is a solution of $(\Sigma, \Gamma, \U)$.
Now the interpretation $I' = \sigma \bullet I$ is given below.
\begin{align*}
&I(c) = \alpha\\
&I(d) = \diamond\\
&I(f)(\alpha) = \diamond, I(f)(\beta) = \diamond, I(f)(\omega) = \diamond\\
&I(P) = \set{(\alpha, \star), (\alpha, \diamond), (\omega, \diamond)}
\end{align*}
This new interpretation is a solution of $(\Sigma, \Gamma, \U')$.

The special case of the Value Relabeling Theorem where $\U = \U'$ is crucial to the understanding of symmetry in MSFMF problems.
In this setting, each bijection $\sigma_\theta$ is a permutation of the set $\U(\theta)$.

\begin{definition}[Domain Permutation]
	Let $P = (\Sigma, \Gamma, \U)$ be a MSFMF problem.
	A \emph{domain permutation} of $P$ is a collection $\sigma$ of permutations that contains a permutation $\sigma_\theta: \U(\theta) \to \U(\theta)$ for each $\theta \in \Theta(\Sigma)$.
	We denote the set of all domain permutations of $P$ by $\DomPerm(P)$.
\end{definition}

The two interpretations $I$ and $\sigma \bullet I$ related by a domain permutation $\sigma$ are both interpretations for the same MSFMF problem.
Moreover, $I$ is a solution of this problem if and only if $\sigma \bullet I$ is a solution of this problem, so checking whether one interpretation satisfies $P$ is equivalent to checking the other.
We call such interpretations \emph{isomorphic}.

\begin{definition}
	Two interpretations $I$ and $I'$ for an MSFMF problem $P$ are \emph{isomorphic} if $I' = \sigma \bullet I$ for some domain permutation $\sigma$ of $P$.
\end{definition}

We restate this special case of the Value Relabeling Theorem using our new terminology for emphasis.

\begin{theorem}[Value Permutation]
	Let $P = (\Sigma, \Gamma, \U)$ be an MSFMF problem and let $\sigma$ be a domain permutation of $P$.
	For any interpretation $I$ of $P$, $I \models \Gamma$ if and only if $(\sigma \bullet I) \models \Gamma$.
	Hence $I$ is a model of $P$ if and only if $(\sigma \bullet I)$ is a model of $P$.
\end{theorem}

For example, consider the same problem $(\Sigma, \Gamma, \U)$ and interpretation $I$ as in the previous example.
\begin{align*}
&\Theta(\Sigma) = \set{A, B}, \mathscr F(\Sigma) = \set{c: A, d: B, f: A \to B}, \mathscr R(\Sigma) = \set{P: A \times B \to \Bool} \\
&\Gamma = \set{\forall x : B . \, P(c, x), \exists y : A . \, f(y) = d } \\
&\U(A) = \set{a_1, a_2, a_3}, \U(B) = \set{b_1, b_2} \\
&I(c) = a_1\\
&I(d) = b_2\\
&I(f)(a_1) = b_2, I(f)(a_2) = b_2, I(f)(a_3) = b_2\\
&I(P) = \set{(a_1, b_1), (a_1, b_2), (a_3, b_2)}
\end{align*}
After acting on $I$ with the domain permutation $\sigma$ defined by
\begin{align*}
& \sigma_A(a_1) = a_3, \sigma_A(a_2) = a_1, \sigma_A(a_3) = a_2 \\
& \sigma_B(b_1) = b_2, \sigma_B(b_2) = b_1,
\end{align*}
the resulting interpretation $I' = (\sigma \bullet I)$ is given as follows.
\begin{align*}
&I(c) = a_3\\
&I(d) = b_1\\
&I(f)(a_3) = b_1, I(f)(a_1) = b_1, I(f)(a_2) = b_1\\
&I(P) = \set{(a_3, b_2), (a_3, b_1), (a_2, b_1)}
\end{align*}
As expected, both $I$ and $I'$ are solutions.

Since each domain permutation relates $I$ to a different isomorphic interpretation, there are many essentially equivalent solutions to a problem, related by these permutations of values.
It would make sense to call these domain permutations ``value symmetries" for MSFMF.
We will avoid this for now since later we wish to relate them to our established definition of value symmetries in CSP problems.

It is worth noting a couple of properties of domain permutations.
These trivially fall out of the definitions.

\begin{definition}[Composition of Domain Permutations]
	Let $P = (\Sigma, \Gamma, \U)$ be a MSFMF problem.
	The \emph{composition} $\sigma \circ \gamma$ of two domain permutations $\sigma$ and $\gamma$ is defined as the domain permutation whose permutation on $\U(\theta)$ is $\sigma_\theta \circ \gamma_\theta$ for each $\theta \in \Theta(\Sigma)$.
\end{definition}

\begin{proposition}
	For any MSFMF problem $P$, $\DomPerm(P)$ forms a group under composition.
	The identity domain permutation is the domain permutation comprised of the identity map on $\U(\theta)$ for each sort $\theta$.
\end{proposition}

\begin{proposition}
	Let $P$ be a MSFMF problem.
	For any domain interpretations $\sigma, \pi \in \DomPerm(P)$ and interpretation $I$,
	\begin{itemize}
		\item $(\sigma \circ \gamma) \bullet I = \sigma \bullet (\gamma \bullet I)$, and
		\item $\id \bullet I = I$,
	\end{itemize}
	where $\id$ is the identity domain permutation.
\end{proposition}

The following corollary is immediate from the above two propositions.

\begin{corollary}
	Let $P$ be a MSFMF problem.
	The action of $\DomPerm(P)$ on the set of interpretations is a group action.
\end{corollary}

\subsection{Framework of Extended MSFMF Problems}

Our ultimate goal is to unify these various definitions of symmetry in finite model finding and Cohen et al.'s symmetry framework for CSPs.
First we will relate the definitions of symmetry in finite model finding just within the context of finite model finding, then consider how they integrate into Cohen et al.'s symmetry framework.

While not strictly allowed in first-order logic, it is useful to allow domain elements as terms.
For example, if a MSFMF problem $P$ contains a sort $A$, a functional symbol $f: A \to A$, and the domain assignment $\U$ has $\U(A) = \set{a_1, \dotsc, a_n}$, it would be allowed to write $\forall x: A, \, f(x) \neq x \lor x = a_3$ as a formula, with the domain element $a_3$ appearing as a term of sort $A$.
To make this distinct from ``pure" first-order logic where this is prohibited, we call such terms \emph{extended}.

\begin{definition}[Extended Many-Sorted Finite Model Finding Problem]
	Let $\Sigma$ be a signature and $\U$ a domain assignment for $\Sigma$.
	An \emph{extended term} is constructed in the same way as terms, but for each sort $\theta$ it allows every domain value $v \in \U(\theta)$ to be an extended term of sort $\theta$.
	An \emph{extended formula} is defined analogously to formulas, but with extended terms rather than terms.
	Similarly, an \emph{extended many-sorted finite model-finding problem} (eMSFMF), an \emph{interpretation} thereof, and the \emph{satisfies} relation are all defined analogously to their pure MSFMF counterparts.
	Semantically, domain elements evaluate to themselves.
\end{definition}

There are many reasons to allow extended problems.
As we will see later, it can be used to add symmetry breaking constraints without straying far from first-order logic, and also facilitates an easier discussion on how to combine symmetry breaking strategies.
Additionally, extended terms can be used to represent partial interpretations as in Torlak's work \cite{torlak_constraint_2009} or ground formulas that have expanded their quantifiers as in Zhang's work \cite{zhang_constructing_1996}.
Solvers such as Falcon \cite{zhang_constructing_1996} and Fortress \cite{fitzgerald_finite_2016} treat domain elements as terms for convenience and simplicity.

It is extremely important to note that the Value Permutation theorem does not necessarily hold when considering an extended MSFMF problem, since values can occur within the formula in ways that make them non-interchangeable.
Thus, in eMSFMF problems it no longer suffices to consider all domain permutations as being symmetries (as in pure problems), but rather some subset of these permutations are symmetries.
Guided by the approach taken by Cohen et al. \cite{cohen_symmetry_2006}, we start with a very general definition of symmetry as those domain permutations that preserve solutions.

\begin{definition}[Domain Symmetry, Isomorphic Interpretations]
	Let $P = (\Sigma, \Gamma, \U)$ be an eMSFMF problem.
	A domain permutation $\sigma$ of $P$ is a \emph{domain symmetry} if for every interpretation $I$, $(\sigma \bullet I) \models \Gamma$ exactly when $I \models \Gamma$.
	We will denote the set of all domain symmetries of $P$ by $\DomSym(P)$.
	Two interpretations $I$ and $I'$ of $P$ are \emph{isomorphic} if $I' = \sigma \bullet I$ for some domain symmetry $\sigma$ of $P$.
\end{definition}

\begin{proposition}
	Let $P$ be an eMSFMF problem.
	The set $\DomSym(P)$ of domain symmetries forms a group under function composition.
\end{proposition}

The Value Permutation Theorem can be rephrased as stating that in a pure MSFMF problem, every domain permutation is a domain symmetry.
Thus, domain symmetries encompass the definitions of symmetry found in the works of Peltier \cite{peltier_new_1998} and Claessen and S\"orenson \cite{claessen_new_2003}.
It also matches the definition of value symmetry by Torlak \cite{torlak_constraint_2009} as permutations of the values that preserve solutions.

Next, we define a set of symmetries that preserves constraints as a subgroup of the solution-preserving symmetries.

\begin{definition}[Action of a Domain Permutation on Formulas]
	Let $P = (\Sigma, \Gamma, \U)$ be an eMSFMF problem, and let $\sigma$ be a domain permutation of $P$.
	We define the action of $\sigma$ on a formula $\Phi \in \Gamma$ as producing the formula $\sigma \bullet \Phi$ obtained by replacing each occurrence of all domain elements in $\Phi$ by applying the appropriate permutation in $\sigma$.
	That is, where a domain element $d$ appears in $\Phi$ as a term of sort $\theta$, in $\sigma \bullet \Phi$ this term is replaced with $\sigma_\theta(d)$.
	We define the action of $\sigma$ on $\Gamma$ by $\sigma \bullet \Gamma = \displaystyle\bigcup_{\Phi \in \Gamma} \sigma \bullet \Phi$.
\end{definition}

\begin{definition}[Constraint Domain Symmetry]
	Let $P = (\Sigma, \Gamma, \U)$ be an eMSFMF problem.
	A domain permutation $\sigma$ of $P$ is a \emph{constraint domain symmetry} if $(\sigma \bullet \Gamma) = \Gamma$.
	We will denote the set of all constraint domain symmetries of $P$ by $\ConDomSym(P)$.
\end{definition}

Since partial instantiations and ground formulas can be represented by extended terms, constraint domain symmetries capture the definitions of Zhang \cite{zhang_constructing_1996}, Audemard and Benhamou \cite{audemard2001symmetry}, and Baumgartner et al  \cite{baumgartner_computing_2009}.
So these definitions, domain symmetry and constraint domain symmetry, are sufficiently general to encompass the previous definitions of symmetry in finite model finding literature.
Like Cohen et al., we can view the constraint symmetries as a subgroup of the more general class of symmetries.
The following theorem will be concluded as a corollary later.

\begin{theorem}
	\label{theorem-con-dom-sym-subgroup}
	For any eMSFMF problem $P$, $\ConDomSym(P)$ is a subgroup of $\DomSym(P)$.
\end{theorem}

\subsection{Unused and Interchangeable Values}

Now that we have this rigorous notion of domain symmetry, it is worth considering one class of easily detectable domain symmetries.
We mentioned that some domain permutations may not be domain symmetries since values may appear in formulas in ways that make them non-interchangeable with other values.
This suggests that any domain permutation that acts as the identity on all values that appear in the formulas is a domain symmetry.

\begin{definition}
	Let $P = (\Sigma, \Gamma, \U)$ be an eMSFMF problem.
	A domain permutation $\sigma$ is \emph{occurrence-fixing} if, for every sort $\theta$ and value $v \in \U(\theta)$, if $v$ appears in $\Gamma$ as a term of sort $\theta$ then $\sigma_\theta(v) = v$ (i.e. it does not change values appearing in the formulas).
\end{definition}

It follows trivially from the definition that the set of occurrence-fixing domain permutations is a subgroup of the constraint domain symmetries.
Since they leave all values that occur in the formulas unchanged, applying them to the formulas leaves the formulas unchanged.

\begin{proposition}
	\label{prop-occurence-fixing-subgroup}
	For any eMSFMF problem $P$, the set of occurrence-fixing domain permutations of $P$ is a subgroup of $\ConDomSym(P)$.
\end{proposition}

Another way to think about this proposition is that values that do not appear in the formulas are ``interchangeable".

\begin{definition}[Solely Permute, Value-Interchangeable Set]
	Let $(\Sigma, \Gamma, \U)$ be an eMSFMF problem and $\theta \in \Theta(\Sigma)$.
	A domain permutation $\sigma$ \emph{solely permutes} a subset $X \subseteq \U(\theta)$ if $\sigma_\theta(v) = v$ for all $v \in \U(\theta) \setminus X$ and $\sigma_{\theta'} = \id_{\U(\theta')}$ for all sorts $\theta' \neq \theta$ (i.e. $\sigma$ acts as the identity everywhere but on $X$).
	A set of $X \subseteq \U(\theta)$ is \emph{value-interchangeable} for $\theta$ if any domain permutation $\sigma$ of $P$ that solely permutes $X$ is a domain symmetry.
\end{definition}

So Proposition \ref{prop-occurence-fixing-subgroup} can be rephrased as stating that for any sort $\theta$, the set of values in $\U(\theta)$ that do not appear in formulas is value-interchangeable.
Depending on the specific $\Gamma$, there may be more sets of interchangeable values.

Recall that domain symmetries can be composed together and yield another symmetry.
Therefore we can compose together symmetries that each solely permute a specific set of values to yield more complicated symmetries.

\begin{proposition}
	\label{prop-value-interchange}
	Let $(\Sigma, \Gamma, \U)$ be an eMSFMF problem.
	For each sort $\theta$, let $X_\theta$ be a value-interchangeable set for $\theta$.
	If $\sigma$ is a domain permutation such that, for every sort $\theta$ and $v \in \U(\theta) \setminus X_\theta$, $\sigma_\theta(v) = v$, then $\sigma$ is a domain symmetry.
\end{proposition}

\subsection{FMF Symmetry in the CSP Framework}

With these definitions of domain symmetry and constraint domain symmetry to unify the symmetry definitions throughout finite model finding, we now move to relate these notions of symmetry to symmetries in constraint satisfaction problems.
A natural approach is to view finite model finding as a constraint satisfaction problem via a transformation between the two kinds problems.

One such transformation is specified by Zhang \cite{zhang_constructing_1996}.
The idea is to view a function-argument pair $f(a_1, \dotsc, a_n)$, where $f$ is a functional symbol and $a_1, \dotsc, a_n$ are appropriately sorted values, as a single \emph{cell} that needs to be assigned a value.
Thus we can make these function-argument pairs into the variables of the CSP.
Essentially, this CSP formulation forgets that there are functions and relations at all, and instead turns each function-argument or relation-argument pair into a variable, for which an appropriate value is searched.

\begin{definition}[Flat CSP]
	Let $P = (\Sigma, \Gamma, \U)$ be an eMSFMF problem.
	The \emph{flat CSP} of $P$ is the CSP $(X, D, C)$ constructed as follows.
	\begin{itemize}
		\item For each functional symbol $f: A_1 \times \dotsb \times A_n \to B$ and values $a_1 \in \U(A_1), \dotsc, a_n \in \U(A_n)$, $X$ contains the function-argument pair $f(a_1, \dotsc, a_n)$ as a variable, and $D(f(a_1, \dotsc, a_n)) = \U(B)$.
		\item For each relational symbol $R: A_1 \times \dotsb \times A_n \to \Bool$ and values $a_1 \in \U(A_1), \dotsc, a_n \in \U(a_n)$, $X$ contains the relation-argument pair $R(a_1, \dotsc, a_n)$ as a variable, and $D(R(a_1, \dotsc, a_n)) = \Bool$.
		\item $C$ contains one constraint $C_\phi$ for each formula $\phi$.
		The scope of $C_\phi$ consists of the function-argument and relation-argument pairs whose function or relation occurs in $\phi$.
		$C_\phi$ allows an assignment for its scope if and only if $C_\phi$ evaluates to $\True$ under this assignment.
	\end{itemize}
\end{definition}

Unfortunately, this translation does not preserve the intuitive notion of ``value symmetry".
Recall that CSP symmetries are permutations on variable-value pairs.
We would hope that domain symmetries in an MSFMF problem would correspond to value symmetries in its CSP, since domain symmetries arise from permutations of values.
However, consider an MSFMF problem $(\Sigma, \Gamma, \U)$ that has a function symbol $f: A \to A$ and a domain symmetry $\sigma$.
Recall how $\sigma$ acts on an interpretation $I$: for $a \in \U(A)$, if $I(f)(a) = y$ then $I(f)( \sigma_A(a) ) = \sigma_A(y)$.
Therefore the natural way to create a variable-value permutation $\sigma^*$ from $\sigma$ would be to have $\sigma^*(f(a), y) = (f(\sigma_A(a)), \sigma_A(y))$.
This is not necessarily a purely value symmetry since it may change the variable in the binding.

We propose an alternative transformation from an eMSFMF problem to a CSP.
The key insight is that while it is usually thought that the domains for variables in a CSP are simple, atomic objects like integers, there is nothing restricting them from being more complicated objects like functions and sets.
Rather than forgetting that a functional symbol $f: A \to B$ represents a function and treating function-argument pairs as individual variables, this transformation makes $f$ into a variable and treats its possible values as functions from $\U(A) \to \U(B)$.
This makes complete assignments of the CSP exactly into interpretations of the eMSFMF problem, and vice versa.

\begin{definition}[Functional CSP]
	Let $P = (\Sigma, \Gamma, \U)$ be an eMSFMF problem.
	For each functional symbol $f: A_1 \times \dotsb \times A_n \to B$, let $D_f$ consist of the set of functions from domain $\U(A_1) \times \dotsb \times \U(A_n)$ to codomain $\U(B)$.
	For each relational symbol $R: A_1 \times \dotsb \times A_n \to \Bool$, let $D_R$ consist of the set of $n$-ary relations $\mathcal R \subseteq \U(A_1) \times \dotsb \times \U(A_n)$.
	The \emph{functional CSP} of $P$ is the CSP $P_{\FCSP} = (X, D, C)$ where
	\begin{enumerate}
		\item $X$ is the set of functional and relational symbols of $\Sigma$,
		\item $D$ assigns each functional symbol $f$ the set of possible values $D_f$ and assigns each relational symbol $R$ the set of possible values $D_R$, and
		\item $C$ contains one constraint $C_\phi$ for each formula $\phi$, where the scope of $C_\phi$ is consists of the functional symbols appearing in $\phi$ and $C_\phi$ allows an assignment for its scope if and only if $C_\phi$ evaluates to $\True$ under this assignment.
	\end{enumerate}
\end{definition}

We now define how to view a domain permutation as a permutation on variable-value bindings in the functional CSP.
Note that this functional extension of a domain permutation leaves the CSP variables fixed, so this transformation does not have the same limitation as the flat CSP transformation.

\begin{definition}[Functional Extension of a Domain Permutation]
	Let $P = (\Sigma, \Gamma, \U)$ be an eMSFMF problem let $\sigma$ be a domain permutation of $P$.
	The \emph{functional extension} $\sigma^\F$ of $\sigma$ is a permutation $\sigma^\F: \Bindings(P_{\FCSP}) \to \Bindings(P_{\FCSP})$ defined as follows.
	
	Let $f: A_1 \times \dotsb \times A_n \to B$ be a functional symbol and $F: \U(A_1) \times \dotsb \times \U(A_n) \to \U(B)$ a concrete function.
	$\sigma^\F$ maps the binding $(f, F)$ to the binding $(f, F')$, where $F': \U(A_1) \times \dotsb \times \U(A_n) \to \U(B)$ is defined by
	\[ F'(a_1, \dotsc, a_n) = \sigma_B( F( \sigma_{A_1}^{-1}(a_1), \dotsc, \sigma_{A_n}^{-1}(a_n) ) ) .\]
	Equivalently, if $F(a_1, \dotsc, a_n) = b$, then $F'(\sigma_{A_1}(a_1), \dotsc, \sigma_{A_n}(a_n)) = \sigma_B(b)$.
	
	Let $R: A_1 \times \dotsb \times A_n \to \Bool$ be a relational symbol and $\R \subseteq \U(A_1) \times \dotsb \times \U(A_n)$ a concrete relation.
	$\sigma^\F$ maps the binding $(R, \R)$ to the binding $(R, \R')$, where $\R' \subseteq \U(A_1) \times \dotsb \times \U(A_n)$ is defined by
	\[ (a_1, \dotsc, a_n) \in \R \iff ( \sigma_{A_1}(a_1), \dotsc, \sigma_{A_n}(a_n) ) \in \R' .\]
	
	For a set $X$ of domain permutations, we will let $X^\F$ denote the set of functional extensions of all domain permutations in $X$.
\end{definition}

An interpretation of an MSFMF problem is a complete assignment to the functional CSP, and vice versa.
We do not mean that one merely induces another, but the interpretation itself is both kinds of mathematical objects.
Moreover, an interpretation satisfies the MSFMF problem if and only if it is a solution to the functional CSP.
It is plainly visible how the functional extension is designed to mimic the action of the domain permutation on an interpretation.
The following proposition follows easily from those respective definitions.

\begin{proposition}
	\label{prop-functional-extension}
	Let $P = (\Sigma, \Gamma, \U)$ be an eMSFMF problem let $\sigma$ be a domain permutation of $P$.
	Let $I$ be an interpretation of $\sigma$.
	We have that $\sigma \bullet I = \sigma^\F \bullet I$.
	That is, the interpretation $\sigma \bullet I$ created by acting on $I$ by $\sigma$ (treating $I$ as an interpretation and $\sigma$ as a domain permutation) is the same object as the complete assignment $\sigma^\F \bullet I$ created by acting on $I$ using the functional extension of $\sigma$ (treating $I$ as a complete assignment to the functional CSP and $\sigma^\F$ as a permutation on the bindings of the CSP).
\end{proposition}

We get the following corollary.
\begin{corollary}
	\label{corollary-dom-sym-sol-sym}
	Let $P = (\Sigma, \Gamma, \U)$ be an eMSFMF problem.
	If $\sigma$ is a domain symmetry of $P$, then $\sigma^\F$ is a solution symmetry of $P_{\FCSP}$.
\end{corollary}
\begin{proof}
	Let $I$ be a solution for $P_{\FCSP}$ (and hence also an interpretation of $P$).
	By Proposition \ref{prop-functional-extension}, $\sigma^\F \bullet I = \sigma \bullet I$.
	Since $\sigma$ is a domain symmetry, $\sigma \bullet I$ satisfies $P$, and hence $\sigma^\F \bullet I = \sigma \bullet I$ is a solution to the functional CSP.
	Therefore $\sigma^F$ is a solution symmetry.
\end{proof}

The following two results also follow quickly from the definitions. 

\begin{lemma}
	\label{lemma-functional-group-axioms}
	Let $P = (\Sigma, \Gamma, \U)$ be an eMSFMF problem.
	If $\sigma$ and $\pi$ are domain symmetries of $P$, then
	\begin{itemize}
		\item $\sigma^\F \circ \pi^\F = (\sigma \circ \pi)^\F$
		\item $(\sigma^\F)^{-1} = (\sigma^{-1})^\F$
	\end{itemize}
\end{lemma}

\begin{theorem}
	For a MSFMF instance $P$, $\DomSym(P)^\F$ is a subgroup of $\SolSym(P_{\FCSP})$.
\end{theorem}
\begin{proof}
	Corollary \ref{corollary-dom-sym-sol-sym} and Lemma \ref{lemma-functional-group-axioms} together show that $\DomSym(P)^\F$ is closed under composition and inverses.
	$\DomSym(P)^\F$ is non-empty since $\DomSym(P)$ is non-empty.
	These are all the requirements for a subgroup.
\end{proof}

Additionally, since constraint domain symmetries preserve the set of constraints of the MSFMF problem, they also preserve the hyperedges in the microstructure complement of the functional CSP.
Therefore we conclude the following, and as promised have Theorem \ref{theorem-con-dom-sym-subgroup} as a corollary.
\begin{theorem}
	For a MSFMF instance $P$, $\ConDomSym(P)^\F$ is a subgroup of $\ConSym(P_{\FCSP})$.
\end{theorem}

The functional CSP transformation allows us to view the domain symmetries of a finite model finding problem as a subgroup of the solution symmetries for the functional CSP, thereby showing how symmetries in finite model finding and CSP symmetries relate to each other.
This transformation also gives us a vehicle to discuss variable symmetries and variable-value symmetries in finite model finding (as those respective kinds of symmetries for the functional CSP).
Such symmetries are not examined in current FMF literature, so this suggests further areas of research.
While $\DomSym(P)^\F$ is a subgroup of $\SolSym(P_{\FCSP})$, the containment may be strict.
There may be more solution symmetries to the functional CSP than there are domain symmetries.
Non-trivial variable symmetries of the functional CSP are one example.
In the next section we will see a concrete example of even some value symmetries that do not correspond to domain symmetries.
A further exploration of the difference between these two sets is warranted, and may enable us to have a greater understanding of symmetries in finite model finding.

\section{Effect of Sorts on Symmetries}
\label{section-sorts}

Most presentations of first-order logic differ from ours in that they are concerned with single-sorted (also called unsorted) first-order logic.
After all, sorts can be simulated using predicates, so there appears at first glance to be no theoretical need for them.
However as Claessen and S\"orenson \cite{claessen_new_2003} note, sort information improves search time.
Many widely used tools for automated theorem proving, such as SMT solvers \cite{barrett_smt-lib_nodate}, use a many-sorted system.
In this section, we first review work by Claessen and S\"orenson \cite{claessen_new_2003} linking many-sorted and single-sorted systems.
We then greatly expand on their work to provide new insights on how sorts relate specifically to the existence of symmetries and come to a new understanding of the process of sort inference.

\subsection{Claessen and S\"orenson's Insights}
Claessen and S\"orenson \cite{claessen_new_2003} were primarily concerned with single-sorted finite model finding, but were interested in how the information acquired from sort inference could be used to help finite model finding for originally single-sorted problems.

Their observation was how to link single-sorted and many-sorted problems.
We change the language of their presentation to match our notation, but the ideas are the same.
Consider a single-sorted finite model finding problem $P = (\Sigma, \Gamma, \U)$ where $\Sigma$ contains just one sort $U$ and $\U(U) = D$.
It may possible to replace occurrences of the sort symbol $U$ within $\Sigma$ and $\Gamma$ with new sort symbols $U_1, \dotsc, U_n$ and still obtain a well-formed signature $\Sigma'$ from $\Sigma$ and set $\Gamma'$ of $\Sigma'$-terms from $\Gamma$.

For example, consider the following single-sorted problem.
\begin{align*}
&\Theta(\Sigma) = \set{U}, \mathscr F(\Sigma) = \set{c_1: U, c_2: U, f: U \to U}, \mathscr R(\Sigma) = \emptyset \\
&\Gamma = \set{f(c_1) \neq c_2, \forall x : U . \, f(x) \neq c_2} \\
&\U(U) = D
\end{align*}
It is possible to replace some occurrences of $U$ with new sorts $U_1$ and $U_2$ to obtain the following signature $\Sigma'$ and set $\Gamma'$ of $\Sigma'$-terms.
\begin{align*}
&\Theta(\Sigma') = \set{U_1, U_2}, \mathscr F(\Sigma') = \set{c_1: U_1, c_2: U_2, f: U_1 \to U_2}, \mathscr R(\Sigma') = \emptyset \\
&\Gamma' = \set{f(c_1) \neq c_2, \forall x : U_1 . \, f(x) \neq c_2}
\end{align*}

After obtaining $\Sigma'$ and $\Gamma'$, a domain assignment can be obtained by setting $\U'(U_i) = D$ for $i = 1, \dotsc, n$.
That is, choose each sort to use the same domain as the original problem.
The key observation is that any interpretation $I$ of the original single-sorted problem is an interpretation of the many-sorted problem.
Moreover, $I$ satisfies $(\Sigma, \Gamma, \U)$ if and only if $I$ satisfies $(\Sigma, \Gamma, \U')$.
Therefore one can infer sorts from the original single-sorted problem and then solve the many-sorted problem instead.
Sort inference is analogous to type inference in programming languages, and standard algorithms such as Hindley-Milner type inference can be used.

Claessen and S\"orenson describe two benefits to inferring sorts.
The first is for symmetry breaking using constants, which we will discuss in a later section.
The latter, which they call \emph{sort size reduction}, is useful but not related to symmetry and so not discussed here.

\subsection{New Insights}

While Claessen and S\"orenson were interested in inferring sorts for single-sorted problems, a similar process can be done for problems already containing multiple sorts.
An even more general sorting can be found.
Let us formalize this notion.
We could not find literature formalizing sort inference specifically for FOL, so we adapt terminology from Pierce's book on type systems in programming languages \cite{pierce_types_2002}.

\begin{definition}[Sort Substitution]
	A sort substitution $\eta$ is a finite mapping from sort symbols to sort symbols, usually written in the form $\set{A_1 \mapsto B_1, \dotsc, A_n \mapsto B_n}$ for sort symbols $A_1, \dotsc A_n, B_1, \dotsc, B_n$, where each $A_i$ is distinct.
	For a given sort $S$, if $(S \mapsto S') \in \eta$ then $\eta(S) = S'$.
	Otherwise $\eta(S) = S$.
	A sort substitution $\eta$ acts on a signature $\Sigma$ by producing a new signature $\eta\Sigma$ obtained by applying $\eta$ to each sort $S$ occurring in $\Sigma$.
	A sort substitution $\eta$ acts on a set $\Gamma$ of $\Sigma$-formulas by producing a set $\eta\Gamma$ of $\eta\Sigma$-formulas obtained by applying $\eta$ to each sort $S$ occurring in $\Gamma$.
\end{definition}

For example, consider the following signature and set of formulas.
\begin{align*}
&\Theta(\Sigma) = \set{A, B, C}, \mathscr F(\Sigma) = \set{c_1: A, c_2: B, c_3: C, f: A \times B \to C}, \mathscr R(\Sigma) = \set{P: A \times B \to \Bool} \\
&\Gamma = \set{\forall x : A . \forall y : B . \, \forall z : C .\, f(x, y) = z}
\end{align*}
Applying the sort substitution $\eta = \set{A \mapsto D, B \mapsto A}$ yields the signature $\Sigma' = \eta \Sigma$ and set of formulas $\Gamma' = \eta \Gamma$ defined as follows.
\begin{align*}
&\Theta(\Sigma') = \set{D, A, C}, \mathscr F(\Sigma') = \set{c_1: D, c_2: A, c_3: C, f: D \times A \to C}, \mathscr R(\Sigma') = \set{P: D \times A \to \Bool} \\
&\Gamma' = \set{\forall x : D . \, \forall y : A . \, \forall z : C .\, f(x, y) = z}
\end{align*}


Now we can formalize the idea of a more general sorting.

\begin{definition}[Less Specifically or More Generally Sorted]
	Let $\Sigma, \Sigma'$ be signatures, and let $\Gamma, \Gamma'$ be sets of formulas over $\Sigma$ and $\Sigma'$ respectively.
	We say that $(\Sigma', \Gamma')$ is \emph{less specifically sorted}, or \emph{more generally sorted}, than $(\Sigma, \Gamma)$ if there exists a type substitution $\eta$ such that $\Sigma = \eta\Sigma'$ and $\Gamma = \eta\Gamma'$.
	If such a relationship holds between $\Sigma, \Sigma', \Gamma, \Gamma'$ and a type substitution $\eta$, we write $(\Sigma', \Gamma') \sqsubseteq_\eta (\Sigma, \Gamma)$.
\end{definition}

As an example, define $(\Sigma, \Gamma)$ and $(\Sigma', \Gamma')$ as follows.
\begin{align*}
&\Theta(\Sigma) = \set{A, B}, \mathscr F(\Sigma) = \set{c : A, f: A \times B \to A}, \mathscr R(\Sigma) = \emptyset \\
&\Gamma = \set{\forall x : A . \, \forall y : B .\, f(x, y) \neq c} \\
&\Theta(\Sigma') = \set{A, B, C}, \mathscr F(\Sigma') = \set{c : C, f: A \times B \to C}, \mathscr R(\Sigma') = \emptyset \\
&\Gamma' = \set{\forall x : A . \, \forall y : B .\, f(x, y) \neq c}
\end{align*}
$(\Sigma', \Gamma')$ is less specifically sorted than $(\Sigma, \Gamma)$ since for the sort substitution $\eta = \set{C \mapsto A}$, we have that $(\eta \Sigma', \eta \Gamma') = (\Sigma, \Gamma)$.
Sort inference is the process of going from $(\Sigma, \Gamma)$ to a less specifically sorted $(\Sigma', \Gamma')$.

Now the question remains how to deal with domain assignments when performing sort inference.
Recall how during sort inference for the single-sorted case, when the universal sort $U$ was split into different sorts $U_1, \dotsc, U_n$, we took $\U'(U_i) = \U(U)$ for each $i$.
Now there are possibly many sorts in the original problem with different domains, but a similar approach can be taken.
During sort inference, each sort is split into one or more sorts.
In the example above, going from $(\Sigma, \Gamma)$ to the less specifically sorted $(\Sigma', \Gamma')$ splits the sort $A$ into two sorts, $A$ and $C$, since $\eta(A) = \eta(C) = A$.
Now, if a sort $A$ is split into sorts $A_1, \dotsc, A_n$ by sort inference, we take $\U'(A_i) = \U(A)$.
Once again this means that any interpretation $I$ of $P = (\Sigma, \Gamma, \U)$ is an interpretation of $P' = (\Sigma', \Gamma', \U')$ and vice versa.
If under $I$ some term $t'$ of sort $A_i$ occurring in $P'$ evaluates to a value in $\U(A_i)$, then in $P$ the corresponding term $t = \eta t'$, which is of sort $A$, also evaluates to a value in $\U(A_i) = \U(A)$.
Importantly, since $t'$ and $t$ differ only by sort annotations that play no role in evaluation, $t'$ and $t$ evaluate to the exact same value under $I$.
Therefore, each pair of correspond formulas $\Phi$ and $\Phi'$ evaluate to the same truth value under $I$.
We formalize this concept and then summarize the above discussion as a theorem.

\begin{definition}[Less Specifically or More Generally Sorted MSFMF Problem]
	Let $P = (\Sigma, \Gamma, \U)$ be an MSFMF problem.
	We say that an MSFMF problem $P' = (\Sigma', \Gamma', \U')$ is \emph{less specifically sorted}, or \emph{more generally sorted}, than $P$ if all of the following hold.
	\begin{itemize}
		\item $(\Sigma', \Gamma') \sqsubseteq_\eta (\Sigma, \Gamma)$ for some sort substitution $\eta$, and 
		\item For each sort $S \in \Theta(\Sigma)$, each sort $S' \in \set{S' \in \Theta(\Sigma') : \eta(S') = S}$ satisfies that $\U'(S') = \U(S)$.
		That is, if $\eta$ maps $S'$ to $S$ then they share the same domain in their respective problems.
	\end{itemize}
	If such a relationship holds for a specific $\eta$, we write $P' \sqsubseteq_\eta P$.
\end{definition}

\begin{theorem}
	\label{theorem-sort-inference}
	Let $P = (\Sigma, \Gamma, \U)$ and $P' = (\Sigma', \Gamma', \U')$ be MSFMF problems such that $P' \sqsubseteq_\eta P$ for some $\eta$.
	Let $I$ be an interpretation of $P$ (and hence also of $P'$).
	Then for every formula $\Phi' \in \Gamma'$ and its counterpart $\Phi = \eta \Phi' \in \Gamma$, $I \models \Phi'$ if and only if $I \models \Phi$.
\end{theorem}

An easy corollary of this theorem is that instead of solving a given MSFMF problem, we can instead solve a more generally sorted version of the same problem.
As will be seen, this is beneficial to symmetry breaking.
However, it also suggests something deeper about how sort systems relate to symmetries.

A more generally sorted version of a problem may contain more domain symmetries than the original problem.
Consider the example below.
\begin{align*}
&\Theta(\Sigma) = \set{A, B}, \mathscr F(\Sigma) = \set{c : A, f: A \times B \to A}, \mathscr R(\Sigma) = \emptyset \\
&\Gamma = \set{\forall x : A . \, \forall y : B .\, f(x, y) \neq c} \\
&\U(A) = \set{1, \dotsc, k_A}, \U(B) = \set{1, \dotsc, k_B} \\
&\Theta(\Sigma') = \set{A, B, C}, \mathscr F(\Sigma') = \set{c : C, f: A \times B \to C}, \mathscr R(\Sigma') = \emptyset \\
&\Gamma' = \set{\forall x : A . \, \forall y : B .\, f(x, y) \neq c} \\
& \U(A) = \set{1, \dotsc, k_A}, \U(B) = \set{1, \dotsc, k_B}, \U(C) = \set{1, \dotsc, k_A}
\end{align*}
Here $P' = (\Sigma', \Gamma', \U') \sqsubseteq_\eta P = (\Sigma, \Gamma, \U)$, where $\eta = \set{C \mapsto A}$.
All domain permutations of $P$ are domain symmetries (since the problem is pure), and there are exactly $k_a! k_b!$ such domain permutations.
Similarly all domain permutations of $P'$ are domain symmetries, though there are $k_a! k_b! k_a!$ of them.
Therefore $P'$ has more domain symmetries than $P$.
This is at first seems to suggest that the sort inference process introduced symmetries.

However, a closer look at Theorem \ref{theorem-sort-inference} shows that this thinking is misleading.
It is true that $P'$ has more domain symmetries than $P$.
However, consider the functional CSPs of $P'$ and $P$.
Because of how the domain assignments are related, functions and relations have the same domains in the CSPs of $P$ and $P'$.
Moreover, since corresponding formulas evaluate to the same truth values under the same interpretations, the two CSPs also have the exact same constraints.
That is to say, the CSPs of the two problems are exactly the same, and hence share the same set of solution symmetries.
Therefore sort inference did not ``create" symmetries, but instead ``moved" some from the more abstract set of solution symmetries to the set of domain symmetries.
Alternatively, sort inference can be viewed as a \emph{symmetry detection} mechanism.
The presence of sorts in the problem then \emph{certifies} the existence of these symmetries.

In summary, the CSP framework developed earlier has allowed us to make two key realizations.
First, there are more value symmetries to a MSFMF problem than the domain symmetries.
Second, sort inference is a mechanism for detecting more abstract symmetries and subsequently rewriting the problem so that they become domain symmetries.

\section{Static Symmetry Breaking}
\label{section-static-sym}

Symmetries divide the search space into equivalence classes based on the orbits of the group action on interpretations or sets of bindings.
For a given equivalence class, its members are either entirely solutions or entirely non-solutions.
This suggests a significant search optimization, since only one member from each equivalence class need be tested to determine whether the problem has a solution.
A search that tests the entire search space is redundant.
In practice it is difficult to fully eliminate this redundancy and test just one member from each class, but there are practical ways of reducing the redundancy and avoiding searching isomorphic interpretations.
This process is called \emph{symmetry breaking}.

Gent et al. \cite{gent_symmetry_2006} describe three main approaches for symmetry breaking for CSP.
The first is \emph{reformulation}, whereby the abstract problem is re-written as a different equivalent problem in such a way to reduce the amount of symmetry present.
However this is difficult to automate.
The second is to add apply symmetry breaking during the search process \emph{dynamically}.
The third tactic, pioneered by Crawford et al. \cite{crawford_symmetry-breaking_1996}, is to add extra constraints to the base problem in the hopes of preventing the solver from exploring redundant parts of the search space.
As long as the constraints are satisfied by at least one member of each equivalence class, this strategy of \emph{static} symmetric breaking is sound.

These three general strategies have been used in finite model-finding as well.
SEM-style model finders that directly explore the search space use dynamic symmetry breaking techniques.
MACE-style model finders, which reduce the model finding problem to another problem (such as SAT) and then invoke an external solver, use static symmetry breaking techniques.
Our focus for this thesis is on static symmetry breaking.
The seminal strategies for static symmetry breaking in finite model finding are those developed by Claessen and S\"orenson \cite{claessen_new_2003} for their Paradox solver.
Their techniques are used for example by Reger et al. \cite{creignou_finding_2016} in their experiments with the Vampire theorem prover and also Vakili and Day \cite{fitzgerald_finite_2016} in Fortress.
Claessen and S\"orenson do not prove the soundness of their technique since they claim it is intuitive, but for greater understanding and to facilitate development of more techniques we believe that a rigorous mathematical proof is valuable.
Indeed, more nuances are at work with their symmetry breaking technique for single-sorted functions than it first appears.
In this section we will present their symmetry breaking techniques and prove their correctness.
In later sections we move to present new symmetry breaking techniques.

\subsection{Existing Techniques for Constants}

First, Claessen and S\"orenson \cite{claessen_new_2003} propose a symmetry breaking scheme for constants.
Consider a pure finite model-finding problem $(\Sigma, \Gamma, \U)$ that has two constants $c_1, c_2, \dotsc, c_n$ of sort $A$, and suppose $\U(A) = \set{a_1, a_2, \dotsc, a_m}$.
Suppose there exists a satisfying interpretation $I$ for this problem (which we do not necessarily have), and we wish to construct our own satisfying interpretation.
It is always possible to apply a domain symmetry to obtain an interpretation $I'$ where $I'(c_1) = a_1$.
Therefore it would not affect the satisfiability of the problem to add the constraint $c_1 = a_1$.
In this new interpretation, either $c_2$ is assigned the same value as $c_1$ or a different value.
If it is the same value then obviously $c_2$ is assigned $a_1$.
If it is a different value, since $a_2, \dotsc, a_m$ are interchangeable, we can again apply a domain symmetry that swaps $I''(c_2)$ and $a_2$, thus obtaining an interpretation $I''$ where $I''(c_2) = a_2$ while leaving $I''(c_1) = a_1$.
Therefore, we could initially have added the constraints $c_1 = a_1$ and $c_2 = a_1 \lor c_2 = a_2$ without affecting satisfiability.
Similarly, the constraint $c_3 = a_1 \lor c_3 = a_2 \lor c_3 = a_3$ could be added, reflecting that either $c_3$ is the same as $c_1$ or $c_2$ and hence is $a_1$ or $a_2$, or it is different and we can arbitrarily choose it to be $a_3$ since the remaining values are interchangeable.
In general, the symmetry breaking constraint
\[
	\bigvee_{i = 1}^k c_k = a_i
\]
can be added for $k = 1, \dotsc, \min\set{m, n}$.
We call these constraints the \emph{ordered constant constraints}.
That is, $c_k$ must be assigned one of the first $k$ values.
Additionally, it can also be constrained that if $c_k$ is assigned value $a_d$, one of the earlier constants must have been assigned $a_{d - 1}$.
This is done by adding what Claessen and S\"orenson call \emph{canonicity constraints}.
For each $k = 2, \dotsc, \min\set{m, n}$ and $d = 2, \dotsc, k$, they add the constraint
\[
	(c_k = a_d) \implies \left( \bigvee_{i = 1}^{k - 1} c_i = a_{d - 1} \right).
\]
Claessen and S\"orenson added these constraints at the SAT level during translation, but we can treat them as extended terms.

We will prove the soundness of this scheme.
To do this, we must show that at least one member of each equivalence class of interpretations satisfies these constraints.
Equivalently, given an arbitrary interpretation $I$, we must show that there exists an isomorphic interpretation $I'$ of $I$ that satisfies these constraints.

\begin{lemma}
	\label{lemma-nice-permutation0}
	Let $X$ be a set of size $n$ and let $x_1, \dotsc, x_m$ be elements of $X$ that are not necessarily distinct.
	There exists a bijection $\sigma : X \to \set{1, \dotsc, n}$ such that, for each $i = 1, \dotsc, m$,
	\begin{itemize}
		\item $\sigma(x_i) \le i$, and
		\item for all $d = 2, \dotsc, i$, if $\sigma(x_i) = d$ then, for some $j \in \set{1, \dotsc, i - 1}$, $\sigma(x_j) = d - 1$.
	\end{itemize}
\end{lemma}
\begin{proof}
	Consider the subsequence $(z_1, \dotsc, z_k)$ obtained from $(x_1, \dotsc, x_m)$ by selecting only the first occurence of each distinct element $x$ and maintaining the same order.
	For example, if the original sequence is $(4, 2, 2, 3, 1, 4, 3)$ then select the subsequence $(4, 2, 3, 1)$.
	
	It suffices to show that there exists a bijection $\sigma: X \to \set{1, \dotsc, n}$ such that, for each $i = 1, \dotsc, k$,
	\begin{itemize}
		\item $\sigma(z_i) \le i$, and
		\item for all $d = 2, \dotsc, i$, if $\sigma(z_i) = d$ then, for some $j \in \set{1, \dotsc, i - 1}$, $\sigma(z_j) = d - 1$.
	\end{itemize}
	To construct such a $\sigma$, simply define $\sigma(z_i) = i$ for $i = 1, \dotsc, k$.
	The rest of the values $\sigma(x)$ for $x \in X \setminus \set{z_1, \dotsc, z_k}$ can then be chosen in any way to complete $\sigma$ to a bijection.
\end{proof}

\begin{theorem}[Soundness of Constant Symmetry Breaking]
	\label{theorem-soundness-constants}
	Let $(\Sigma, \Gamma, \U)$ be a (pure) MSFMF problem with a sort $A$.
	Enumerate $\U(A)$ as $a_1, \dotsc, a_m$.
	Let $c_1, \dotsc, c_n$ be constants of $\Sigma$ of sort $A$.
	Let $I$ be any interpretation of $(\Sigma, \Gamma, \U)$.
	There exists an isomorphic interpretation $I'$ of $I$ such that the following all hold.
	\begin{enumerate}
		\item For $k = 1, \dotsc, \min\set{m, n}$, $I'$ satisfies $\bigvee_{i = 1}^k c_k = a_i$.
		\item For $k = 2, \dotsc, \min\set{m, n}$ and $d = 2, \dotsc, k$, $I'$ satisfies $(c_k = a_d) \implies \left( \bigvee_{i = 1}^{k - 1} c_i = a_{d - 1} \right)$.
	\end{enumerate}
\end{theorem}
\begin{proof}
	Let $r = \min\set{m, n}$.
	An interpretation $I'$ satisfies the above constraints if and only if, for $k = 1, \dotsc, r$, 
	\begin{itemize}
		\item $I'(c_k) \in \set{a_1, \dotsc, a_k}$, and
		\item for all $d = 2, \dotsc, k$, if $I'(c_k) = a_d$ then, for some $i \in \set{1, \dotsc, k - 1}$, $I'(c_i) = a_{d - 1}$.
	\end{itemize}

	Define $y_k = I(c_k)$ for each $c_k$.
	We will construct $I'$ that satisfies the above requirements by applying an appropriate domain symmetry $\sigma$ to $I$.
	For each sort $S$ that is not $A$, choose $\sigma_S$ to be an arbitrary permutation on $\U(S)$.
	For whichever $\sigma_A$ is chosen, we will have that $I'(c_k) = \sigma_A(y_k)$.
	Any choice of $\sigma_A$ makes $\sigma$ a domain symmetry, but we have additional requirements on $\sigma_A$ for $I'$ to satisfy the constraints.
	So we need to construct $\sigma_A$ such that, for $k = 1, \dotsc, r$,
	\begin{itemize}
		\item $\sigma_A(y_k) \in \set{a_1, \dotsc, a_k}$, and
		\item for all $d = 2, \dotsc, k$, if $\sigma_A(y_k) = a_d$ then, for some $i \in \set{1, \dotsc, k - 1}$, $\sigma_A(y_i) = a_{d - 1}$.
	\end{itemize}
	Such a permutation exists by Lemma \ref{lemma-nice-permutation0}.
\end{proof}

While their work is primarily about single-sorted finite model-finding, Claessen and S\"orenson do describe how to modify this approach for the presence of multiple sorts.
These constraints can be applied individually for each sort.
We will describe this in more detail, with a proof of soundness, in Section \ref{section-combining}.

\subsection{Existing Techniques for Single-Sorted Unary Functions}

Next, we present Claessen and S\"orenson's symmetry breaking scheme in the presence of single-sorted unary functions of the form $f: A \to A$.
It seems natural that if we can constrain $c_1 = a_1$, $c_2 = a_1 \lor c_2 = a_2$ and so on up to $c_k = a_1 \lor \dotsb \lor c_k = a_k$, that we could now add the constraints
\begin{align*}
	&f(a_1) = a_1 \lor \dotsb \lor f(a_1) = a_{k + 1} \\
	&f(a_2) = a_1 \lor \dotsc \lor f(a_2) = a_{k + 1} \lor f(a_2) = a_{k + 2} \\
	&\dotsb
\end{align*}
and so on.
This is the scheme that Claessen and S\"orenson use.
However, as they note this scheme only works when $k > 0$ (i.e. there is a constant that has had symmetry breaking constraints added for it).
If there are no constants, Claessen and S\"orenson add an artificial one and add the corresponding symmetry breaking constraints.
They do not elaborate on why they do this, so we will explain why.

Consider if there are no constants and we wish to add symmetry breaking constraints for a function $f: A \to A$.
Our intuition for why we could choose $a_1$ to be the value assigned to a constant $c$ was that all of the values were ``the same" as far as $c$ was concerned.
However things are more complicated when trying to assign a value for $f(a_1)$.
While choosing $I(f)(a_1) = a_2$ is not fundamentally different than choosing $I(f)(a_1) = a_3$, choosing $I(f)(a_1) = a_1$ is fundamentally different than choosing $I(f)(a_1) = a_2$.
For example, $I(f)(a_1) = a_1$ satisfies $\forall x : A . \, f(x) = x$ while $I(f)(a_1) = a_2$ does not.
The reason that the proof for Theorem \ref{theorem-soundness-constants} fails to generalize to such functions is because of how a domain permutation $\sigma$ acts on an interpretation.
Recall that if $I(f)(a_1) = y$ then $(\sigma \bullet I)(f)( \sigma_A(a_1) ) = \sigma_A(y)$.
The domain permutation must be applied to inputs as well as outputs.
If $I(f)(a_1) = a_1$, it does not suffice to simply assign $\sigma_A(a_1) = a_2$ since then we have an interpretation $I'$ with $I'(f)(a_2) = a_2$, which doesn't necessarily satisfy that $I'(f)(a_1) = a_1$.
In fact if $I(f)(a_i) \neq a_i$ for all $i$ then there is no domain permutation such that $(\sigma \bullet I)(f)(a_1) = a_1$.
We did not face this problem for constants because nullary functions have no inputs that are permuted.

Nonetheless, the values $a_2, \dotsc, a_n$ are intuitively ``the same" as far as $f(a_1)$ is concerned.
It only matters whether the output is the equal to or different than the input.
Then for $f(a_2)$, either $f(a_2)$ is $a_1$ or $a_2$, or it is one of the remaining unused values $a_3, \dotsc, a_n$, which again appear to be equivalent choices for $f(a_2)$.
All unused values that are distinct from the input are equivalent choices.
This suggests we should be able to add the following constraints.
\begin{align*}
& f(a_1) = a_1 \lor f(a_1) = a_2 \\
& f(a_2) = a_1 \lor f(a_2) = a_2 \lor f(a_2) = a_3 \\
&\dotsb
\end{align*}

We call these constraints \emph{ordered range constraints}.
We now prove the soundness of this symmetry breaking scheme for unary functions.
It is more involved than our first soundness proof, since constructing an appropriate symmetry is non-trivial.

\begin{lemma}
	\label{lemma-nice-permutation1}
	Let $X$ be a set of size $n$ and let $h$ be any function from $X$ to $X$.
	There exists a bijection $\sigma: X \to \set{1, \dotsc, n}$ such that, for each $v \in X$, if $\sigma(v) = j$ then $\sigma(h(v)) \le j + 1$.
\end{lemma}

We find it helps to interpret this lemma graphically.
Let $D$ be a directed graph on node set $X$ such that each vertex $x$ has one outgoing arc with head $h(x)$.
This lemma says that there is a way to order the vertices such that the head of each arc goes no further than one vertex to the right of the tail.

\begin{proof}[Proof of Lemma \ref{lemma-nice-permutation1}]
	To prove the existence of such a permutation, it is equivalent to show that there is an ordering $(v_1, \dotsc, v_n)$ of the elements of $X$ (where $v_i \neq v_j$ for $i \neq j$) such that
	$h(v_j) \in \set{v_1, \dotsc, v_{j + 1}}$ for each $j \in \set{1, \dotsc, n - 1}$.
	To do this, we will show by induction that for $k = 1, \dotsc, n$ there exists a sequence $(v_1, \dotsc, v_k)$ of distinct elements of $X$ such that $h(v_j) \in \set{v_1, \dotsc, v_{j + 1}}$  for each $j \in \set{1, \dotsc, k - 1}$.
	We will say such a sequence is a \emph{$k$-good sequence}.
	
	First suppose $k = 1$.
	The required property for a sequence to be $1$-good is vacuously true.
	Therefore we can select an arbitrary vertex $v_1$ to obtain the $1$-good sequence $(v_1)$.
	
	Now suppose $2 \le k \le n$.
	By the inductive hypothesis there exists a $(k - 1)$-good sequence $(v_1, \dotsc, v_{k - 1})$.
	It remains to make a choice for $v_k$ from $X \setminus \set{v_1, \dotsc, v_{k - 1}}$ and to show that $h(v_{k - 1}) \in \set{v_1, \dotsc, v_k}$, since then $(v_1, \dotsc, v_k)$ is a $k$-good sequence.
	There are two cases depending on whether $h(v_{k - 1})$ lies in the $(k - 1)$-good sequence.
	If $h(v_{k - 1}) \in \set{v_1, \dotsc, v_{k - 1}}$, then any choice of $v_k$ from $X \setminus \set{v_1, \dotsc, v_{k - 1}}$ will satisfy that $h(v_{k - 1}) \in \set{v_1, \dotsc, v_k}$, so $v_k$ can be chosen arbitrarily.
	If $h(v_{k - 1}) \notin \set{v_1, \dotsc, v_{k - 1}}$, then selecting $v_k = h(v_{k - 1})$ satisfies that $h(v_{k - 1}) \in \set{v_1, \dotsc, v_k}$.
	Hence we can construct a $k$-good sequence.
\end{proof}

\begin{theorem}[Soundness of Ordered Range Constraints For Single-Sorted Unary Functions]
	Let $(\Sigma, \Gamma, \U)$ be a (pure) MSFMF problem and let $f: A \to A$ be a functional symbol of $\Sigma$.
	Enumerate $\U(A)$ as $a_1, \dotsc, a_m$.
	Let $I$ be any interpretation of $(\Sigma, \Gamma, \U)$.
	There exists an isomorphic interpretation $I'$ of $I$ that satisfies the ordered range constraints $\bigvee_{j =1}^{i + 1} f(a_i) = a_j$ for each $i = 1, \dotsc, m - 1$.
\end{theorem}
\begin{proof}
	An interpretation $I'$ satisfies the ordered range constraints if and only if $I'(f)(a_i) \in \set{a_1, \dotsc, a_{i + 1}}$ for $i = 1, \dotsc, m - 1$.
	
	For convenience let $h(a_i) = I(f)(a_i)$ for each $a_i$.
	We will construct $I'$ by applying an appropriate domain symmetry $\sigma$ to $I$.
	For every sort $S$ that is not $A$, choose $\sigma_S$ to be an arbitrary permutation on $\U(S)$.
	For whichever $\sigma_A$ is chosen, the newly constructed interpretation $I'$ is such that $I'(f)(\sigma_A(a)) = \sigma_A(h(a))$ for each $a \in \U(A)$.
	We need to construct a permutation $\sigma_A : \U(A) \to \U(A)$ such that, for each $a \in \U(A)$ and $j \in \set{1, \dotsc, m - 1}$ if $\sigma_A(a) = a_j$ then $\sigma(h(a)) \in \set{a_1, \dotsc, a_{j + 1}}$.
	Such a permutation exists by Lemma \ref{lemma-nice-permutation1}.                                                                                                                                                                                                                                                                                                                                              
\end{proof}

This symmetry breaking technique does not generalize easily to higher arity functions.
To get any symmetry reduction out of a binary function $g: A \times A \to A$, Claessen and S\"orenson define a new unary function $h: A \to A$, add the constraint $\forall x : A . \, h(x) = f(x, x)$, and then perform symmetry breaking using the function $h$.

\section{New Symmetry Breaking for Functions}
\label{section-functions}

Although the symmetry breaking techniques introduced by Claessen and S\"orenson \cite{claessen_new_2003} are very effective, they have some limitations.
First, we discussed how extra values are needed for the ordered range constraints when considering unary functions as opposed to constants.
Additionally, it was mentioned that their symmetry breaking approach does not generalize easily to higher arity functions.
While their paper discusses how to add symmetry breaking constraints for constants in a many-sorted setting, its discussion of functions is limited to a single-sorted setting.
The question remains whether more can be done in a many-sorted setting.
In this section, we will present new symmetry breaking constraints for multi-sorted functions, as well as rigorously prove the correctness of this new symmetry breaking scheme.
We will show that it generalizes very well to a class of higher arity functions.

\subsection{Multi-Sorted Unary Functions}

Recall that the problem we faced for a function $f: A \to A$ when attempting to constraint $f(a_1)$ is that because symmetries must be applied consistently to both the arguments and result of a function, we cannot just constrain $f(a_1) = a_1$.
However if we consider a multi-sorted unary function symbol $f: A \to B$, this problem vanishes.
We can separately permute the input and output values because they belong to different sorts.
The domain permutation can leave the input sort fixed, and only permute the output sort.
The proof of Theorem \ref{theorem-soundness-constants} now can be easily generalized, and we can add the following \emph{strong ordered range constraints}.

\begin{align*}
& f(a_1) = b_1 \\
& f(a_2) = b_1 \lor f(a_2) = b_2 \\
& f(a_3) = b_1 \lor f(a_3) = b_2 \lor f(a_3) = b_3 \\
&\dotsb
\end{align*}

\begin{theorem}[Soundness of Strong Ordered Range Constraints For Multi-Sorted Unary Functions]
	\label{theorem-soundness-sorts-unary}
	Let $(\Sigma, \Gamma, \U)$ be a (pure) MSFMF problem and let $f: A \to B$ be a functional symbol of $\Sigma$, where $A$ and $B$ are distinct sort symbols.
	Enumerate $\U(A)$ as $a_1, \dotsc, a_m$ and $\U(B)$ as $b_1, \dotsc, b_n$.
	Let $I$ be any interpretation of $(\Sigma, \Gamma, \U)$.
	There exists an isomorphic interpretation $I'$ of $I$ that satisfies the strong ordered range constraints $\bigvee_{j = 1}^{i} f(a_i) = b_j$ for each $i = 1, \dotsc, \min \set{m, n}$.
\end{theorem}
\begin{proof}
	Let $r = \min \set{m, n}$.
	An interpretation $I'$ satisfies the strong ordered range constraints if and only if $I'(f)(a_i) \in \set{b_1, \dotsc, b_i}$ for $i = 1, \dotsc, r$.
	
	Define $y_i = I(f)(a_i)$ for each $a_i$.
	We will construct $I'$ by applying an appropriate domain symmetry $\sigma$ to $I$.
	For each sort $S$ that is not $B$, select $\sigma_S$ to be the identity on $\U(S)$.
	In particular, choose $\sigma_A$ to be the identity on $\U(A)$.
	Now for whichever $\sigma_B$ is chosen, we will have $I'(f)(a_i) = \sigma_B(y_i)$ for each $i$.
	So we need to construct $\sigma_B$ such that $\sigma_B(y_i) \in \set{b_1, \dotsc, b_i}$ for $i = 1, \dotsc, r$.
	The existence of such a $\sigma_B$ follows from Lemma \ref{lemma-nice-permutation0}.
\end{proof}

\subsection{Higher Arity Functions}
Now that we know more symmetry breaking can be done for functions of the form $f: A \to B$, the question remains what to do about higher arity functions.
Fortunately, this new symmetry breaking strategy can be easily generalized to higher arities.
What made the proof of Theorem \ref{theorem-soundness-sorts-unary} so simple was that we could leave the inputs sorts fixed and simply permute the output sorts.
Therefore if a function symbol does not have its output sort as one of its input sorts, we can apply a similar strategy.

\begin{definition}
	A function symbol $f: A_1 \times \dotsb \times A_k \to B$ is \emph{domain-range distinct (DRD)} if its result type $B$ is distinct from each of its argument types $A_1, \dotsc, A_k$ (though $A_1, \dotsc A_k$ need not be distinct).
\end{definition}

\begin{theorem}[Soundness of Strong Ordered Range Constraints For DRD Functions]
	\label{theorem-soundness-sorts-drd}
	Let $(\Sigma, \Gamma, \U)$ be a (pure) MSFMF problem and let $f: A_1 \times \dotsb \times A_k \to B$ be a DRD function symbol of $\Sigma$.
	Take an arbitrary ordering $t_1, \dotsc, t_m$ of the tuples in $\U(A_1) \times \dotsb \times \U(A_k)$.
	Also enumerate $\U(B)$ as $b_1, \dotsc b_n$.
	Let $I$ be any interpretation of $(\Sigma, \Gamma, \U)$.
	There exists an isomorphic interpretation $I'$ of $I$ that satisfies the strong ordered range constraints $\bigvee_{j = 1}^{i} f(t_i) = b_j$ for each $i = 1, \dotsc, \min \set{m, n}$.
\end{theorem}
\begin{proof}
	Let $r = \min \set{m, n}$.
	An interpretation $I'$ satisfies the strong ordered range constraints if and only if $I'(f)(t_i) \in \set{b_1, \dotsc, b_i}$ for $i = 1, \dotsc, r$.
	
	Define $y_i = I(f)(t_i)$ for each $t_i$.
	We will construct $I'$ by applying an appropriate domain symmetry $\sigma$ to $I$.
	For each sort $S$ that is not $B$, select $\sigma_S$ to be the identity on $\U(S)$.
	In particular, choose $\sigma_{A_i}$ to be the identity on $\U(A_i)$ for each $A_i$.
	Now for whichever $\sigma_B$ is chosen, we will have $I'(f)(t_i) = \sigma_B(y_i)$ for each $i$.
	So we need to construct $\sigma_B$ such that $\sigma_B(y_i) \in \set{b_1, \dotsc, b_i}$ for $i = 1, \dotsc, r$.
	The existence of such a $\sigma_B$ follows from Lemma \ref{lemma-nice-permutation0}.
\end{proof}

Consider the problem of finding a $9 \times 9$ Latin Square.
One way to formulate this as a finite model finding problem is by having a single sort $N$ with values $\U(N) = \set{1, 2, \dotsc, 9}$ along with a function symbol $f: N \times N \to N$ that represents what entry is assigned to each row-column pair, and then adding appropriate constraints on $f$ for columns and rows.
Without greater knowledge of the problem, not much symmetry breaking can be justified using Claessen and S\"orenson's technique, since $f$ is a binary function that uses only a single sort.
However, sort inference easily reveals that there is no reason that rows, columns, and grid entries need to use the same sort.
Instead, distinct sorts $R, C, E$ could be used with $\U(R) = \U(C) = \U(E) = \set{1, 2, \dotsc, 9}$.
The function can instead become the domain-range distinct function $f: R \times C \to E$, and so we can justify adding the following strong ordered range constraints.
\begin{align*}
&f(1, 1) = 1 \\
&f(1, 2) = 1 \lor f(1, 2) = 2 \\
&f(1, 3) = 1 \lor f(1, 3) = 2 \lor f(1, 3) = 3 \\
&\dotsb \\
&f(1, 9) = 1 \lor f(1, 9) = 2 \lor f(1, 9) = 3 \lor \dotsb \lor f(1, 9) = 9
\end{align*}
This demonstrates the advantage that having a many-sorted system and sort inference can provide to symmetry breaking.

\section{New Symmetry Breaking for Predicates}
\label{section-predicates}

In this section we present a new symmetry breaking technique that we have developed to reduce symmetries in the presence of predicates.
Occasionally in this section we use \emph{cycle} notation for a permutation.
When $(a_1, a_2, \dotsc, a_n)$ is written, it means the permutation that sends $a_i$ to $a_{i + 1}$ for $i = 1, \dotsc, n-1$, sends $a_n$ to $a_1$, and leaves any other values of the set in question fixed.

\subsection{Unary Predicates}
Consider a signature $\Sigma$ with a sort $A$ and unary predicate symbol $P: A \to \Bool$ as well as a domain assignment $\U$ that assigns $\U(A) = \set{a, b, c, d, e, f}$.
For an interpretation $I$, $I(P)$ is a unary relation $I(P) = \set{\alpha_1, \dots, \alpha_k} \subseteq \U(A)$, which is simply a set of domain elements.
Now consider any pure MSFOL formula $\Phi$ over $\Sigma$.
Assume for now we do not add any other symmetry breaking constraints.

Suppose we know that there is a satisfying interpretation $I$ which assigns $I(P) = \set{d, e, f}$.
There is nothing special about the particular values $d, e, f$.
By using the Value Relabeling Theorem, we could consistently interchange $d, e$ and $f$ with $a, b$, and $c$ respectively in the interpretation to get another interpretation $I'$ that assigns $I'(P) = \set{a, b, c}$ and still satisfies $\Phi$.
Conversely, if we know that there is no satisfying interpretation $I'$ that assigns $I'(P) = \set{a, b, c}$, we know by the Value Relabeling Theorem that there is no satisfying interpretation $I$ that assigns $I'(P) = \set{d, e, f}$.
The same is true not just for $\set{a, b, c}$, but any set of three elements of $\U(A)$.

The key insight here is that there because values are interchangeable, any choice for $I(P)$ of size $k$ is just as good as any other choice for $I(P)$ of size $k$.
If we choose a specific $k$-subset $Z$ of $\U(A)$ and learn that there is no interpretation that assigns $I(P) = Z$, there is no need to try any other $k$-subset; we know none of them will yield a solution.
We need only try assigning $I(P)$ to be $\emptyset, \set{a}, \set{a, b}, \set{a, b, c}, \set{a, b, c, d}, \set{a, b, c, d, e}$,  $\set{a, b, c, d, e, f}$ to determine whether a solution exists.

We can enforce restricting the search this way by introducing \emph{predicate membership constraints}.
To come up with these constraints, it is perhaps more intuitive to think of constructing the set $I(P)$, and what restrictions we place on the other elements given that a specific element is in the set.
If we wish to put the element $b$ to the set, it means that we must include $a$ in the set, so we add the constraint $P(b) \implies P(a)$.
Next if we wish to add the element $c$, it means that we must have include both $a$ and $b$ in the set, so we add the constraint $P(c) \implies P(a) \land P(b)$.
We proceed similarly for the remaining elements to get the following constraints.

\begin{align*}
&P(b) \implies P(a) \\
&P(c) \implies P(a) \land P(b) \\
&P(d) \implies P(a) \land P(b) \land P(c) \\
&P(e) \implies P(a) \land P(b) \land P(c) \land P(d) \\
&P(f) \implies P(a) \land P(b) \land P(c) \land P(d) \land P(e)
\end{align*}

These constraints are satisfied if and only if $I(P)$ is one of $\emptyset, \set{a}, \set{a, b}, \set{a, b, c}, \set{a, b, c, d}, \set{a, b, c, d, e}$, or $\set{a, b, c, d, e, f}$.
We can simplify these constraints to the following.
\begin{align*}
&P(b) \implies P(a) \\
&P(c) \implies P(b) \\
&P(d) \implies P(c) \\
&P(e) \implies P(d) \\
&P(f) \implies P(e)
\end{align*}

Now we prove the soundness of adding these predicate membership constraints.

\begin{theorem}
	[Soundness of Predicate Membership Constraints for Unary Predicates]
	\label{theorem-soundness-predicate-unary}
	Let $(\Sigma, \Gamma, \U)$ be a (pure) MSFMF problem and let $Q: A \to \Bool$ be a predicate symbol of $\Sigma$.
	Order $\U(A)$ as $a_1, \dotsc, a_m$.
	Let $I$ be any interpretation of $(\Sigma, \Gamma, \U)$.
	There exists an isomorphic interpretation $I'$ of $I$ that satisfies the predicate membership constraints $Q(a_i) \implies Q(a_{i - 1})$ for each $i = 2, \dotsc, m$.
\end{theorem}
\begin{proof}
	An interpretation $I'$ satisfies the required predicate membership constraints if and only if $I'(Q)$ is one of the following sets: $\emptyset, \set{a_1}, \set{a_1, a_2}, \dotsc, \set{a_1, a_2, \dotsc, a_m}$.
	
	Write $I(Q)$ as $\set{y_1, y_2, \dotsc, y_k}$ for some integer $k \ge 0$ (note that the $y_i$'s are distinct).
	We will construct $I'$ by applying an appropriate domain symmetry $\sigma$ to $I$.
	Choose $\sigma_S$ to be an arbitrary permutation of $\U(S)$ for any sort $S \neq A$.
	Now, for whichever $\sigma_A$ is chosen, $I'(Q) = \set{\sigma_A(y_1), \dotsc, \sigma_A(y_k)}$.
	Thus it suffices to show that there exists a permutation $\sigma_A: \U(A) \to \U(A)$ such that $\set{\sigma_A(y_1), \dotsc, \sigma_A(y_k)} = \set{a_1, a_2, \dotsc, a_k}$.
	We can easily construct such a permutation by defining $\sigma_A(y_i) = a_i$ for $i = 1, \dotsc, k$ and $\sigma_A(\alpha) = \alpha$ for all $\alpha \in \U(A) \setminus I(Q)$.
\end{proof}

\subsection{Binary Predicates}

This symmetry breaking strategy is surprisingly difficult generalize to higher arity predicates.
We might intuitively think that since all values are interchangeable, all argument tuples ought to be ``interchangeable" somehow as well.
Then we could order the possible argument tuples as $t_1, \dotsc, t_m$ and define predicate membership constraints $Q(t_i) \implies Q(t_{i - 1})$ for $i = 2, \dotsc, m$.
This intuition is not correct, for two reasons.

Consider a predicate $R: A \times A \to \Bool$ and $\U(A) = \set{a, b, c}$. 
It is wrong to say that any pair in $\U(A) \times \U(A)$ is interchangeable with any other pair in $\U(A) \times \U(A)$.
While we might consider $(a, b)$ to be interchangeable with $(c, d)$, it seems incorrect to consider $(a, a)$ and $(a, b)$ interchangeable because there is a relationship between the first and second elements of the pair $(a, a)$ that is not shared by the elements of the pair $(a, b)$.

As a concrete example, consider the formula $\Phi \coloneqq \forall x : A . \lnot R(x, x)$.
This formula is satisfied if the interpretation $I(R) = \set{(a, b)}$ but is not satisfied if $I(R) = \set{(a, a)}$.
Determining whether there is a satisfying interpretation that assigns a specific set of one tuple to $R$ does not immediately tell us whether there exists a satisfying interpretation that assigns any other specific set of one tuple to $R$.
Examining the proof of Theorem \ref{theorem-soundness-predicate-unary} illuminates where the error lies.
There is no permutation $\sigma_A: \U(A) \to \U(A)$ that can map  $\set{ (a, b) }$ to $\set{(a, a)}$.
Such a permuation would be required to send both $a$ and $b$ to $a$.

However, we have the benefit of working in a many-sorted system.
If instead we have a predicate $Q: A \times B \to \Bool$, we should be able to avoid the tuple element interdependence that arose above.
Formulas like the $\Phi$ above would not be allowed.
We would be allowed two permutations $\sigma_A : \U(A) \to \U(A)$ and $\sigma_B : \U(B) \to \U(B)$, the former operating on the first tuple element and the latter on the second.
The above problem would be no issue.
Learning whether there is a satisfying interpretation for a specific set of one tuple to $R$ does allow us to immediately determine whether there exists a satisfying interpretation that assigns any given set of one tuple to $R$.

Unfortunately, this fails to address a second complication.
Consider the formula $\Phi \coloneqq \forall x : A, \, y, z : B . \, Q(x, y) \land Q(x, z) \implies (y = z)$, and say $\U(A) = \set{a_1, \dotsc, a_n}$, $\U(B) = \set{b_1, \dotsc, b_m}$.
Here, no satisfying interpretation assigns $Q$ a set of one tuple, so our claim above still holds.
However, consider the sets $\set{ (a_1, b_1), (a_1, b_2) }$ and $\set{ (a_1, b_2), (a_2, b_1)  }$.
Assigning $I(Q) = \set{ (a_1, b_1), (a_1, b_2) }$ does not satisfy $\Phi$, while assigning $I(Q) = \set{ (a_1, b_2), (a_2, b_1) }$ does satisfy $\Phi$.
There do not exist permutations $\sigma_A : \U(A) \to \U(A)$ and $\sigma_B : \U(B) \to \U(B)$ such that  $\set{ (\sigma_A(a_1), \sigma_B(b_1)), (\sigma_A(a_1), \sigma_B(b_2)) } =  \set{ (a_1, b_2), (a_2, b_1) }$, since $\sigma_A$ would need to send $a_1$ to both $a_1$ and $a_2$.
The problem here is that $a_1$ appears twice (but in different tuples) in one of the sets, but only once in the other.
While any two individual tuples are ``interchangeable" in that there can be found a pair of permutations (one for each sort) that maps one tuple to the other, this is not in general true for two \emph{sets} of tuples.

While we have painted a grim picture, we do not mean to suggest that nothing can be done.
Indeed, we should expect some symmetry breaking to be possible for higher-arity predicates, since for example the sets $\set{ (a_1, b_1), (a_1, b_2) }$ and $\set{ (a_1, b_2), (a_1, b_3) }$ are related by a pair of permutations, specifically the identity function $\sigma_A = \id_A$ and the cycle $\sigma_B = (b_1 \ b_2 \ b_3)$.
To break this symmetry, we would want to say that the tuple $(a_1, b_3)$ should only be included only if the tuples $(a_1, b_2)$ and $(a_1, b_1)$ have already been included.
We can take this further, and say that the set of all $b$ such that $(a_1, b) \in I(Q)$ should be one of $\emptyset, \set{b_1}, \set{b_1, b_2}$, or $\set{b_1, b_2, b_3}$, which is captured by the following constraints.

\begin{align*}
&Q(a_1, b_2) \implies Q(a_1, b_1) \\
&Q(a_1, b_3) \implies Q(a_1, b_2) \\
&Q(a_1, b_4) \implies Q(a_1, b_3) \\
&\dotsc
\end{align*}

We call these the \emph{$a_1$-predicate membership constraints}.
This can be thought of as taking the partially applied predicate $Q(a_1, \cdot)$ and creating the unary predicate membership constraints for it.
The proof that this symmetry breaking strategy is sound is almost the same proof as for Theorem \ref{theorem-soundness-predicate-unary}.

\begin{theorem}
	[Soundness of Predicate Membership Constraints for Binary Predicates]
	\label{predicate-soundness-binary}
	Let $(\Sigma, \Gamma, \U)$ be a (pure) MSFMF instance and let $Q: A \times B \to \Bool$ be a predicate symbol of $\Sigma$, where $A$ and $B$ are distinct sorts.
	Enumerate $\U(A)$ as $a_1, \dotsc, a_n$ and $\U(B)$ as $b_1, \dotsc, b_m$.
	Let $I$ be any interpretation of $(\Sigma, \Gamma, \U)$.
	There exists an isomorphic interpretation $I'$ of $I$ such that $I'$ satisfies the $a1$-predicate membership constraints $Q(a_1, b_j) \implies Q(a_1, b_{j - 1})$ for $j = 2, \dotsc, m$ .
\end{theorem}
\begin{proof}
	Given an interpretation $J$, let $\restrict {J(Q)} {a_i} = \set{ b \in \U(B) : (a_i, b) \in I'(Q) }$.
	An interpretation $J$ satisfies the $a1$-predicate membership constraint if and only if $\restrict {J(Q)} {a_1}$ is one of the following sets: $\emptyset$, $\set{b_1}$, $\set{b_1, b_2}$, $\dotsc$, $\set{b_1, b_2, \dotsc, b_m}$.
	
	We will construct $I'$ by applying an appropriate domain symmetry $\sigma$ to $I$.
	For each sort $S$ that is not $B$, select $\sigma_S$ to be the identity on $\U(S)$.
	In particular, choose $\sigma_A$ to be the identity on $\U(A)$.
	Now we need to find is a permutation $\sigma_B : \U(B) \to \U(B)$ such that, applying $\sigma_B$ pointwise to let $\restrict {I(Q)} {a_1}$ yields one of the sets listed above.
	Such a permutation is easy to construct.
	Writing $\restrict {I(Q)} {a_1} = \set{\beta_1, \dotsc, \beta_k}$, define $\sigma_B$ such that $\sigma_B(\beta_i) = b_i$ for $i = 1, \dotsc, k$ and $\sigma_B$ acts as the identity on $\U(B) \setminus \restrict {I(Q)} {a_1}$.
\end{proof}

Note that the above theorem only claims that we can add the predicate membership constraints for a single value of $\U(A)$.
We would hope that we could add predicate membership constraints for more values of $a$, but unfortunately this is not the case.
To do so, we would need to find a single $\sigma_B$ that sends $\restrict {I(Q)} {a_i}$ to one of the sets $\emptyset, \set{b_1}, \set{b_1, b_2}, \dotsc, \set{b_1, b_2, \dotsc, b_m}$ for \emph{every} value of $i$, which is not always possible.
Consider the case where $\U(A) = \set{a_1, a_2}$ and $\U(B) = \set{b_1, b_2}$ with a starting interpretation where $I(Q) = \set{ (a_1, b_1), (a_2, b_2) }$.
There are only two permutations on $\U(B)$, the identity and $(b_1 \ b_2)$.
Applying the former yields the same interpretation $I'(Q) = \set{ (a_1, b_1), (a_2, b_2) }$, which fails $Q(a_2, b_2) \implies Q(a_2, b_1)$, while applying the latter yields $I'(Q) = \set{ (a_1, b_2), (a_2, b_1)}$, which fails $Q(a_1, b_2) \implies Q(a_1, b_1)$.

We also want to emphasize that Theorem \ref{predicate-soundness-binary} might fail for a predicate $R: A \times A \to \Bool$ that uses only one argument sort.
The proof takes advantage of the fact that it can work with two permutations by defining $\sigma_A$ to be the identity and only needing to find a suitable $\sigma_B$.
Such a strategy no longer works when only working with a single permutation.
Consider where $\U(A) = \set{a_1, a_2}$ with a starting interpretation $I(R) = \set{ (a_1, a_2), (a_2, a_1) }$.
Both permutations of $\U(A)$ send $I(R)$ to itself, and this interpretation does not satisfy the constraint $Q(a_1, a_2) \implies Q(a_1, a_1)$.
Once more this demonstrates the value of sort inference, since it can be used to identify circumstances where $R$ could instead be more generally typed and allow for greater symmetry breaking.

\section{Combining Symmetry Breaking Strategies}
\label{section-combining}

We have rigorously proven the soundness of various static symmetry breaking schemes in the presence of constants, functions, and predicates.
They cannot be arbitrarily combined however, since adding symmetry breaking constraints introduces domain values into the formulas, and values may no longer be interchangeable.

For example, consider the MSFMF problem $P = (\Theta, \Gamma, \U)$ defined as follows.
\begin{align*}
&\Theta(\Sigma) = \set{A}, \mathscr F(\Sigma) = \set{c: A}, \mathscr R(\Sigma) = \set{P: A \to \Bool} \\
&\Gamma = \set{\lnot P(c) \land \exists x: A . \, P(x) } \\
&\U(A) = \set{a_1, a_2, a_3}
\end{align*}

This problem is satisfiable.
Using Theorem \ref{theorem-soundness-constants}, we can add the ordered constant constraint $c = a_1$.
Alternatively, using Theorem \ref{theorem-soundness-predicate-unary}, we can add the predicate membership constraints $P(a_3) \implies P(a_2)$ and $P(a_2) \implies P(a_1)$.
However, we cannot add all of these constraints simultaneously, since the resulting problem becomes unsatisfiable.
The issue is that after adding the constraint $c = a_1$, the problem is now an extended MSFMF problem instead of a pure problem, and not all of the values in $\U(A)$ are interchangeable.

Fortunately, after adding a set of symmetry breaking constraints, many values are likely to still be interchangeable.
For example, since $a_2$ and $a_3$ have not been used after adding $c = a_1$, the set $\set{a_2, a_3}$ is still value-interchangeable.
In this section we generalize our symmetry breaking strategies and proofs of correctness to extended problems.
The result of this generalization is that we determine under what conditions these symmetry breaking schemes can be combined soundly.
In this setting, for a given sort $A$, $\U(A)$ may not necessarily be interchangeable, but a subset of $\U(A)$ is value-interchangeable.

\subsection{Theorems and Discussion}

The significance of these theorems may be bogged down by the technical details of the proofs, so first we will simply state the theorems and discuss their implications.
Rigorous proofs are provided in the next subsection.

We begin with a theorem analogous to Theorem \ref{theorem-soundness-constants}, which involved symmetry breaking for constants.
\begin{restatable}[Soundness of Constant Symmetry Breaking, Extended]{restatable-theorem}{constextended}
	\label{theorem-soundness-constants-extended}
	Let $(\Sigma, \Gamma, \U)$ be an eMSFMF problem with a sort $A$.
	Let $X \subseteq \U(A)$ be a value-interchangeable set for $A$, and enumerate $X$ as $a_1, \dotsc, a_m$.
	Let $c_1, \dotsc, c_n$ be constants of $\Sigma$ of sort $A$.
	Let $I$ be any interpretation of $(\Sigma, \Gamma, \U)$.
	There exists an isomorphic interpretation $I'$ of $I$ such that all of the following hold.
	\begin{enumerate}
		\item For $k = 1, \dotsc, \min\set{m, n}$, $I'$ satisfies $\left( \bigvee_{i = 1}^k c_k = a_i \right) \lor \left( \bigvee_{a \in \U(A) \setminus X} c_k = a \right)$.
		\item For $k = 2, \dotsc, \min\set{m, n}$ and $d = 2, \dotsc, k$, $I'$ satisfies $(c_k = a_d) \implies \left( \bigvee_{i = 1}^{k - 1} c_i = a_{d - 1} \right)$.
	\end{enumerate}
\end{restatable}

The first thing of note is that Theorem \ref{theorem-soundness-constants-extended} justifies performing symmetry breaking for constants separately for each sort in a pure MSFMF problem, as done by Claessen and S\"orenson \cite{claessen_new_2003}.
For example, suppose there are constants $x_1: A, x_2: A, y_1: B, y_2: B$, and the domain assignment specifies $\U(A) = \set{a_1, a_2}$ and $\U(B) = \set{b_1, b_2}$.
First, we can add the symmetry breaking constraints $x_1 = a_1$ and $x_2 = a_1 \lor x_2 = a_2$.
In this new extended problem, the entirety of $\U(B)$ is still value-interchangeable for $B$, so we can add the constraints $y_1 = b_1$ and $y_2 = b_1 \lor y_2 = b_2$.

More interesting however is when we want to perform symmetry breaking on constants of a sort that already has some of its domain values in the formulas.
For example, consider the extended MSFMF problem $P = (\Theta, \Gamma, \U)$ defined as follows.
\begin{align*}
&\Theta(\Sigma) = \set{A, B}, \mathscr F(\Sigma) = \set{c_1: A, c_2: A, f: B \to A}, \mathscr R(\Sigma) = \set{P: A \to Bool} \\
&\Gamma = \set{P(a_3), f(b_1) = a_4} \\
&\U(A) = \set{a_1, a_2, a_3, a_4, a_5}, \U(B) = \set{b_1, b_2}
\end{align*}
The set of values $X = \set{a_1, a_2, a_5}$ does not appear in the formulas and so is value-interchangeable.
Therefore Theorem \ref{theorem-soundness-constants-extended} justifies the addition of the following constraints (the parentheses are just there for emphasis to separate which of the disjuncts comes from the interchangeable values and which come from the other values).
\begin{align*}
	&\left(c_1 = a_1 \right) \lor \left(c_1 = a_3 \lor c_1 = a_4 \right) \\
	&\left(c_2 = a_1 \lor c_2 = a_2 \right) \lor \left(c_2 = a_3 \lor c_2 = a_4 \right)
\end{align*}

The following is the generalization of Theorem \ref{theorem-soundness-sorts-drd}, which considered DRD functions.

\begin{restatable}[Soundness of Strong Ordered Range Constraints for DRD Functions, Extended]{restatable-theorem}{drdextended}
	\label{theorem-soundness-sorts-drd-extended}
	Let $(\Sigma, \Gamma, \U)$ be an eMSFMF problem and let $f: A_1 \times \dotsb \times A_k \to B$ be a DRD function symbol of $\Sigma$.
	Let $X \subseteq \U(B)$ be a value-interchangeable set for $B$.
	Take an arbitrary ordering $t_1, \dotsc, t_m$ of the tuples in $\U(A_1) \times \dotsb \times \U(A_k)$.
	Enumerate $X$ as $b_1, \dotsc b_n$.
	Let $I$ be any interpretation of $(\Sigma, \Gamma, \U)$.
	There exists an isomorphic interpretation $I'$ of $I$ that satisfies the strong ordered range constraints $\left( \bigvee_{j = 1}^{i} f(t_i) = b_j \right) \lor \left( \bigvee_{b \in \U(B) \setminus X} f(t_i) = b \right)$ for each $i = 1, \dotsc, \min \set{m, n}$.
\end{restatable}

We wish to emphasize the possibly surprising fact that it does not matter whether the input values to a function are interchangeable or not.
As an example, suppose we have the extended MSFMF problem $P = (\Theta, \Gamma, \U)$ defined by the following.
\begin{align*}
&\Theta(\Sigma) = \set{A, B}, \mathscr F(\Sigma) = \set{x: B, y: B, f: B \times C \to A}, \mathscr R(\Sigma) = \emptyset \\
&\Gamma = \set{x = a_1, f(b_1, c_1) = a_1, f(b_2, c_1) \neq f(b_2, c_3), f(b_1, c_1) = f(b_2, c_2)} \\
&\U(A) = \set{a_1, a_2, a_3, a_4, a_5, a_6}, \U(B) = \set{b_1, b_2}, \U(C) = \set{c_1, c_2}
\end{align*}
While $a_1$ appears in the formulas, the values in the set $X = \set{a_2, a_3, a_4, a_5, a_6}$ do not appear and are value-interchangeable for $A$.
Therefore we can add the following symmetry breaking constraints, even though values of sorts $B$ and $C$ have been used in the formulas (again the parentheses separate the disjuncts arising from interchangeable values from the other values).
\begin{align*}
	&\left( f(b_1, c_1) = a_2 \right) \lor \left( f(b_1, c_1) = a_1 \right) \\
	&\left( f(b_1, c_2) = a_2 \lor f(b_1, c_2) = a_3 \right) \lor \left( f(b_1, c_2) = a_1 \right) \\
	&\left( f(b_2, c_1) = a_2 \lor f(b_2, c_1) = a_3 \lor f(b_2, c_1) = a_4 \right) \lor \left( f(b_2, c_1) = a_1 \right) \\
	&\left( f(b_2, c_2) = a_2 \lor f(b_2, c_2) = a_3 \lor f(b_2, c_2) = a_4 \lor f(b_2, c_2) = a_5 \right) \lor \left( f(b_2, c_2) = a_1 \right)
\end{align*}

Finally we generalize Theorem \ref{theorem-soundness-predicate-unary}.
\begin{restatable}[Soundness of Predicate Membership Constraints for Unary Predicates, Extended]{restatable-theorem}{predextended}
	\label{theorem-soundness-predicate-unary-extended}
	Let $(\Sigma, \Gamma, \U)$ be an eMSFMF problem and let $Q: A \to \Bool$ be a predicate symbol of $\Sigma$.
	Let $X \subseteq \U(A)$ be a value-interchangeable set for $A$.
	Order $X$ as $d_1, \dotsc, d_m$.
	Let $I$ be any interpretation of $(\Sigma, \Gamma, \U)$.
	There exists an isomorphic model $I'$ of $I$ that satisfies the predicate membership constraint $Q(d_i) \implies Q(d_{i - 1})$ for each $i = 2, \dotsc, m$.
\end{restatable}

Consider for example the extended MSFMF problem $P = (\Theta, \Gamma, \U)$ defined as follows.
\begin{align*}
&\Theta(\Sigma) = \set{A}, \mathscr F(\Sigma) = \set{c: A}, \mathscr R(\Sigma) = \set{P: A \to \Bool} \\
&\Gamma = \set{c = a_1, \lnot P(c) \land \exists x: A . \, P(x) } \\
&\U(A) = \set{a_1, a_2, a_3}
\end{align*}
This is the first problem we considered in this section, after we have added the symmetry breaking constraint $c = a_1$.
As mentioned earlier, we are not justified in adding the constraints $P(a_3) \implies P(a_2)$ and $P(a_2) \implies P(a_1)$ since $\U(A)$ is not value-interchangeable for $A$.
However $X = \set{a_2, a_3}$ is value-interchangeable for $A$, so Theorem \ref{theorem-soundness-predicate-unary-extended} justifies the addition of the constraint $P(a_3) \implies P(a_2)$.

\subsection{Soundness Proofs}

Here we recall the theorems from the previous section and provide proofs for them.
It is important to remember that unlike in the proofs of the pure counterparts to these theorems, not all domain permutations are domain symmetries, so value-interchangeability has to be taken into consideration to prove that the constructed domain permutations are domain symmetries.

\constextended*

\begin{proof}
	Let $r = \min\set{m, n}$.
	An interpretation $I'$ satisfies the above constraints if and only if, for $k = 1, \dotsc, r$, 
	\begin{itemize}
		\item $I'(c_k) \in \set{a_1, \dotsc, a_k} \cup (\U(A) \setminus X)$, and
		\item for all $d = 2, \dotsc, k$, if $I'(c_k) = a_d$ then, for some $i \in \set{1, \dotsc, k - 1}$, $I'(c_i) = a_{d - 1}$.
	\end{itemize}
	
	Define $y_k = I(c_k)$ for each $c_k$.
	We will construct $I'$ that satisfies the above requirements by applying an appropriate domain symmetry $\sigma$ to $I$.
	For each sort $S$ that is not $A$, choose $\sigma_S$ to be the identity on $\U(S)$.
	For whichever $\sigma_A$ is chosen, we will have that $I'(c_k) = \sigma_A(y_k)$.
	By definition of value-interchangeability, if $\sigma_A$ acts as the identity on $\U(A) \setminus X$, then $\sigma$ is a domain symmetry.
	So it suffices to construct $\sigma_A$ such that,
	\begin{itemize}
		\item $\sigma_A$ acts as the identity on $\U(A) \setminus X$,
		\item for $k = 1, \dotsc, r$, $\sigma_A(y_k) \in \set{a_1, \dotsc, a_k} \cup (\U(A) \setminus X)$, and
		\item for $k = 1, \dotsc, r$ and $d = 2, \dotsc, k$, if $\sigma_A(y_k) = a_d$ then, for some $i \in \set{1, \dotsc, k - 1}$, $\sigma_A(y_i) = a_{d - 1}$.
	\end{itemize}
	Define $\sigma_A$ to be the union of the identity on $\U(A) \setminus X$ and some as of yet undetermined permutation $\pi$ on $X$.
	This choice satisfies the first requirement, as well as the second and third requirements for any $k$ such that $y_k \in \U(A) \setminus X$ (the third requirement holds in this case since $\sigma(y_k) \notin X$ and the antecedent of the implication is false).
	Therefore we need only construct the permutation $\pi: X \to X$ such that, for any $k$ where $y_k \in X$,
	\begin{itemize}
		\item $\sigma_A(y_k) \in \set{a_1, \dotsc, a_k}$, and
		\item for all $d = 2, \dotsc, k$, if $\sigma_A(y_k) = a_d$ then, for some $i \in \set{1, \dotsc, k - 1}$, $\sigma_A(y_i) = a_{d - 1}$.
	\end{itemize}
	
	Construct a subsequence $(z_1, \dotsc, z_l)$ of $(y_1, \dotsc, y_n)$ by removing all values $y_k$ such that $y_k \notin X$ and keeping the same relative ordering.
	By the construction of the subsequence it suffices to construct $\pi$ such that for $k = 1, \dotsc, l$,
	\begin{itemize}
		\item $\sigma_A(z_k) \in \set{a_1, \dotsc, a_k}$, and
		\item for all $d = 2, \dotsc, k$, if $\sigma_A(z_k) = a_d$ then, for some $i \in \set{1, \dotsc, k - 1}$, $\sigma_A(z_i) = a_{d - 1}$.
	\end{itemize}
	Such a permutation exists by Lemma \ref{lemma-nice-permutation0}.
\end{proof}

\drdextended*

\begin{proof}
	Let $r = \min \set{m, n}$.
	An interpretation $I'$ satisfies the strong ordered range constraints if and only if $I'(f)(t_i) \in \set{b_1, \dotsc, b_i} \cup \left( \U(B) \setminus X \right)$ for $i = 1, \dotsc, r$.
	
	For simplicity, define $y_i = I(f)(t_i)$ for each $t_i$.
	We will construct $I'$ by applying an appropriate domain symmetry $\sigma$ to $I$.
	For each sort $S$ that is not $B$, select $\sigma_S$ to be the identity on $\U(S)$.
	In particular, choose $\sigma_{A_i}$ to be the identity on $\U(A_i)$ for each $A_i$.
	Now for whichever $\sigma_B$ is chosen, we will have $I'(f)(t_i) = \sigma_B(y_i)$ for each $i$.
	Additionally, by definition of value-interchangeability, if $\sigma_B$ acts as the identity on $\U(B) \setminus X$, then $\sigma$ is a domain symmetry.
	Therefore, it suffices to construct $\sigma_B$ such that
	\begin{enumerate}
		\item $\sigma_B$ acts as the identity on $\U(B) \setminus X$, and
		\item for each $i = 1, \dotsc, r$, $\sigma_B(y_i) \in \set{b_1, \dotsc, b_i} \cup \left( \U(B) \setminus X\right)$.
	\end{enumerate}
	
	Define $\sigma_B$ such that it is the union of identity on $\U(B) \setminus X$ and an as of yet undetermined permutation $\pi$ on $X$.
	This choice satisfies the first requirement, and satisfies the second requirement for any $i$ such that $y_i \in \U(B) \setminus X$.
	Therefore we need only construct a permutation $\pi: X \to X$ such that, for any $i$ where $y_i \in X$, we have $\sigma(y_i) \in \set{b_1, \dotsc, b_i}$.
	
	Construct a subsequence $(z_1, \dotsc, z_l)$ of $(y_1, \dotsc, y_m)$ by removing all values $y_i$ such that $y_i \notin X$, and keeping the same relative ordering.
	By the construction of the subsequence, it suffices to construct a permutation $\pi: X \to X$ such that $\sigma(z_i) \in \set{b_1, \dotsc, b_i}$ for $i = 1, \dotsc, l$.
	The existence of such a $\pi$ follows from Lemma \ref{lemma-nice-permutation0}.
\end{proof}

\predextended*

\begin{proof}
	An interpretation $I'$ satisfies the required predicate membership constraints if and only if $I'(Q) \cap X$ is one of the following sets: $\emptyset, \set{d_1}, \set{d_1, d_2}, \dotsc, \set{d_1, d_2, \dotsc, d_m}$.
	
	Again we construct $I'$ by applying an appropriate domain symmetry $\sigma$ to $I$.
	Define $\sigma_S$ to be the identity permutation of $\U(S)$ for any sort $S \neq A$.
	Additionally, define $\sigma_A$ so that it is the union of the identity permutation on $\U(A) \setminus X$ and an as of yet to be determined permutation $\pi : X \to X$.
	By the definition of value-interchangeability, $\sigma$ is a domain symmetry.
	Write $I(Q) \cap X$ as $\set{y_1, y_2, \dotsc, y_k}$.
	It suffices to show that there exists a permutation $\pi: X \to X$ such that $\set{\pi(y_1), \dotsc, \pi(y_k)} = \set{d_1, d_2, \dotsc, d_k}$.
	We can easily construct such a permutation by defining $\pi(y_i) = d_i$ for $i = 1, \dotsc, k$.
\end{proof}

\section{Conclusions}
Let us summarize the accomplishments of this thesis.
We reviewed and unified the various notions of symmetry in constraint satisfaction and finite model finding.
Next, we used this framework to demonstrate how sorts serve as proofs of the existence of symmetries and that sort inference operates as a symmetry detection mechanism, bringing us to a fuller understanding of how sorts and symmetries relate.
We presented existing static symmetry breaking techniques used in FMF, providing proofs of correctness that are omitted in existing literature.
Following this we introduced a new symmetry breaking scheme for domain-range distinct functions, which exist only in the many-sorted setting, as we all a new scheme for symmetry breaking with unary predicates.
These new schemes have the potential to increase the effectiveness of static symmetry breaking and thereby improve solver performance, especially for systems that use sorts.
We additionally proved the soundness of both new schemes.
Hopefully our approach to proving the soundness of symmetry breaking constraints can serve as a guide for future authors looking to develop new symmetry breaking schemes so they need not retreat to appealing to intuition.
Finally, we proved conditions for when symmetry breaking constraints can be combined.

There are many avenues for future work.
We saw in Sections \ref{section-fmf} and \ref{section-sorts} that there are more symmetries of a finite model finding problem than just domain symmetries, and it remains to be seen how to exploit these other symmetries.
Also, while we proved conditions for when symmetry breaking constraints can be combined, it still needs to be investigated what combinations are most optimal at reducing symmetries.
In particular, it would be useful to develop some heuristics and collect empirical results.
Additionally, we aim to implement sort inference and these new symmetry breaking schemes in the Fortress model finder.

\bibliography{zotero-bib.bib,other-bib.bib}
\bibliographystyle{ieeetr}

\end{document}